\documentclass{amsart}
\usepackage[dvips]{graphicx}
\usepackage{amscd}
\usepackage{amsmath}
\usepackage{amsxtra}
\usepackage{amsfonts}
\usepackage{amssymb}
\usepackage[all,dvips,arc,curve,color,frame]{xy}

\newtheorem{theorem}{Theorem}[section]
\newtheorem{corollary}[theorem]{Corollary}
\newtheorem{lemma}[theorem]{Lemma}
\newtheorem{proposition}[theorem]{Proposition}

\theoremstyle{definition}
\newtheorem{definition}[theorem]{Definition}
\newtheorem{remark}[theorem]{Remark}

\newtheorem{example}[theorem]{Example}
\theoremstyle{remark}

\renewcommand{\theclaim}{\textup{\theclaim}}

\newtheorem*{acknowledgements}{Acknowledgements}

\numberwithin{equation}{section}

\def\openone

{\mathchoice

{\hbox{\upshape \small1\kern-3.3pt\normalsize1}}

{\hbox{\upshape \small1\kern-3.3pt\normalsize1}}

{\hbox{\upshape \tiny1\kern-2.3pt\SMALL1}}

{\hbox{\upshape \Tiny1\kern-2pt\tiny1}}}

\makeatletter

\newbox\ipbox

\newcommand{\ip}[2]{\left\langle #1\, , \,#2\right\rangle}
\newcommand{\diracb}[1]{\left\langle #1\mathrel{\mathchoice

{\setbox\ipbox=\hbox{$\displaystyle \left\langle\mathstrut
#1\right.$}

\vrule height\ht\ipbox width0.25pt depth\dp\ipbox}

{\setbox\ipbox=\hbox{$\textstyle \left\langle\mathstrut
#1\right.$}

\vrule height\ht\ipbox width0.25pt depth\dp\ipbox}

{\setbox\ipbox=\hbox{$\scriptstyle \left\langle\mathstrut
#1\right.$}

\vrule height\ht\ipbox width0.25pt depth\dp\ipbox}

{\setbox\ipbox=\hbox{$\scriptscriptstyle \left\langle\mathstrut
#1\right.$}

\vrule height\ht\ipbox width0.25pt depth\dp\ipbox}

}\right. }

\newcommand{\dirack}[1]{\left. \mathrel{\mathchoice

{\setbox\ipbox=\hbox{$\displaystyle \left.\mathstrut
#1\right\rangle$}

\vrule height\ht\ipbox width0.25pt depth\dp\ipbox}

{\setbox\ipbox=\hbox{$\textstyle \left.\mathstrut
#1\right\rangle$}

\vrule height\ht\ipbox width0.25pt depth\dp\ipbox}

{\setbox\ipbox=\hbox{$\scriptstyle \left.\mathstrut
#1\right\rangle$}

\vrule height\ht\ipbox width0.25pt depth\dp\ipbox}

{\setbox\ipbox=\hbox{$\scriptscriptstyle \left.\mathstrut
#1\right\rangle$}

\vrule height\ht\ipbox width0.25pt depth\dp\ipbox}

} #1\right\rangle}

\newcommand{\preal}{\operatorname*{Re}}
\newcommand{\cj}[1]{\overline{#1}}

\newcommand{\bz}{\mathbb{Z}}
\newcommand{\M}{\mathcal{M}}

\newcommand{\br}{\mathbb{R}}
\newcommand{\bc}{\mathbb{C}}
\newcommand{\bt}{\mathbb{T}}
\newcommand{\bn}{\mathbb{N}}

\def\blfootnote{\xdef\@thefnmark{}\@footnotetext}

\input xy
\xyoption{all}
\usepackage{amssymb}

\def\F{\mathcal{F}}

\def\E{\mathcal E}

\def\H{\mathcal{H}}

\def\-{^{-1}}

\def\D{\mathcal{D}}

\def\ty{\emptyset}

\def\W{\mathcal{W}}

\begin{document}

\title[SPECTRAL THEORY FOR DISCRETE LAPACIANS]{SPECTRAL THEORY FOR DISCRETE LAPACIANS}
\author{Dorin Ervin Dutkay}
\blfootnote{Research supported in part by a grant from the National Science Foundation DMS-0704191}
\address{[Dorin Ervin Dutkay] University of Central Florida\\
	Department of Mathematics\\
	4000 Central Florida Blvd.\\
	P.O. Box 161364\\
	Orlando, FL 32816-1364\\
U.S.A.\\} \email{ddutkay@mail.ucf.edu}

\author{Palle E.T. Jorgensen}
\address{[Palle E.T. Jorgensen]University of Iowa\\
Department of Mathematics\\
14 MacLean Hall\\
Iowa City, IA 52242-1419\\}\email{jorgen@math.uiowa.edu}
\thanks{} 

\begin{abstract}
   We give the spectral representation for a class of selfadjoint discrete graph Laplacians $\Delta$, with $\Delta$ depending on a chosen graph $G$ and a conductance function $c$ defined on the edges of $G$. We show that the spectral representations for $\Delta$ fall in two model classes, (1) tree-graphs with $N$-adic branching laws, and (2) lattice graphs. We show that the spectral theory of the first class may be computed with the use of rank-one perturbations of the real part of the unilateral shift, while the second is analogously built up with the use of the bilateral shift. We further analyze the effect on spectra of the conductance function $c$: How the spectral representation of $\Delta$ depends on $c$. 
   
   Using $\Delta_G$, we introduce a resistance metric, and we show that it embeds isometrically into an energy Hilbert space. We introduce an associated random walk and we calculate return probabilities, and a path counting number.
\end{abstract}
\maketitle \tableofcontents
\section{Introduction}\label{intr}

One fascination with operator theory is its connections to other areas such as geometry and discrete analysis: Because of applications to electrical networks, to statistical mechanics, and to fractals (see e.g., \cite{Str06, JoPe08a, JoPr89, JSW94, Pow76, Pow79}), there is a recent increased interest in detailed spectral representation for operators on infinite graphs. In addition to the applications, these connections further suggests a search for more a direct link between, on the one hand, metric geometry of infinite graphs, and on the other, a spectral analysis of associated families of operators. 

   However as we see, classical methods break down in infinite discrete models: With Fourier analysis in classical potential theory, it is often possible to represent Laplace operators by multiplication, and hence realize the spectral representation this way. The analysis then breaks up into the study of discrete and continuous parts. However we show that analogues of this that adapt to the discrete case have strong limitations. New tools from operator theory are needed: For example in the discrete case, typically there is not a natural Fourier duality available, and the graph may not even be associated with a group in a way that facilitates computation of spectral representations. Further, in the study of Laplacians on infinite graphs $G$, there are several Hilbert spaces in the picture. Choices must be made: There are Hilbert space completions of functions on the vertices in $G$, and similarly for functions on the edges in $G$ (Definitions \ref{defi}, \ref{defre1}), the energy Hilbert space. 

   Here we focus on classes of graphs $G$ that require new tools. Our conclusions (Theorem \ref{th3.25}, Corollary \ref{coko1}, and Theorem \ref{thti1}) imply that not only is there a direct connection between the spectrum of the graph Laplacian $\Delta_G$ and the metric geometry of $G$; but this connection carries over to the detailed fine-structure of the {\it multiplicity configurations} for $\Delta_G$. 

   Our proofs rely on a mixture of operator theory (section 3) and complex analysis (section 2). Since we address three different audiences, for the convenience of readers, we have included a few details which may be known to operator theorists but not to graph theorists, and vice versa.

   In this paper we study the operator theory of infinite graphs $G$, and especially a natural family of Laplace operators directly associated with the graph in question. These operators depend not only on $G$, but also on a chosen positive real valued function $c$ defined on the edges in $G$. In electrical network models, the function $c$ will determine a conductance number for each edge $e$; the conductance being the reciprocal of the resistance between the endpoint vertices in the edge. Specifically if $e = (xy)$ connects vertices $x$ and $y$ in $G$, the number $c(e)$ is the reciprocal of the resistance between $x$ and $y$. Hence prescribing a conductance leads to classes of admissible flows in $G$ determined from Ohm's law and Kirchhoff's laws of electrical networks. This leads to a measure of energy directly associated with the graph Laplacian. There are Hilbert spaces $H(G)$ which offer a useful spectral theory, and our main results concern the spectral theory of these operators. In a recent paper \cite{Jor08} it was proved that the graph Laplacians are automatically essentially selfadjoint, i.e., that the associated operator closures are selfadjoint operators in $H(G)$.

   Here we give a spectral analysis of the graph Laplacians. We are motivated by a pioneering paper \cite{Pow76} which in an exciting way applies graphs and resistor networks to a problem in quantum statistical mechanics.

    There are many benefits from having a detailed spectral picture of graph Laplacians:  We get a spectral representation realization of the graph Laplacians, the operators $\Delta_{G,c}$, i.e., a unitarily equivalent form of these operators which may arise in a variety of applications. See e.g., \cite{Arv02, PaSc72}.

 In Corollary \ref{corko3}, we obtain a candidate for Brownian motion (independent increments) on a general graph, and for the tree graphs we show in Proposition \ref{prop4.5} that the increments go like the square root of the distance.

       In the course of the proofs of our main results, we are making use of tools from the theory of unbounded operators in Hilbert space, especially \cite{Voi85}, but also von Neumann's deficiency indices, operator closure, operator domains, operator adjoints; and extensions of Hermitian operators with a dense domain in a fixed complex Hilbert space. References for this material include: \cite{Jor77, Jor78, JoPe00,Nel69, vN31, Sto51}. For analysis on infinite graphs and on fractals, see e.g., \cite{BHS05, CuSt07, DoSn84,DuJo06a, HKK02, Hut81, JoPe98, JKS07, Kig03, Str06}.

\subsection{Technical details}

        Let $G = (G^{(0)}, G^{(1)})$ be an infinite graph, $G^{(0)}$ for vertices, and $G^{(1)}$ for edges. Every $x$ in $G^{(0)}$ is connected to a set $\mbox{nbh}(x)$ of other vertices by a finite number of edges, but points in $\mbox{nbh}(x)$ are different from $x$; i.e, we assume that $x$ itself is excluded from $\mbox{nbh}(x)$; i.e., no $x$ in $G^{(0)}$ can be connected to itself with a single edge. Let $c$ be a conductance function defined on $G^{(1)}$.

     Initially, the graph $G$ will not be directed, but when a conductance is fixed, and we study induced current flows, then these current flows will give a direction to the edges in $G$. But the edges in $G$ itself do not come with an intrinsic direction.

     A theorem from \cite{Jor08} states that the Laplace operator $\Delta = \Delta _c$ is automatically essential selfadjoint. By this we mean that $\Delta$ is defined on the dense subspace $D$ (of all the real valued functions on $G^{(0)}$ with finite support) in the Hilbert space $H = H(G) := l^2(G^{(0)})$. The conclusion is that the closure of the operator $\Delta$ is selfadjoint in $H$, and so in particular that it has a unique spectral resolution, determined by a projection valued measure on the Borel subsets the infinite half-line $\br_+$; i.e., the spectral measure takes values in the projections in the Hilbert space $l^2(G^{(0)})$. We work out the measure.

    In contrast, we note that the corresponding Laplace operator in the continuous case is not essential selfadjoint. This can be illustrated for example with $\Delta = - \frac{d^2}{dx^2}$ on the domain $D$ of consisting of all $C^2$ functions on the infinite half-line $\br_+$ which vanish with their derivatives at the end points. In this case, the Hilbert space is $L^2(\br_+)$.

Our main theorems are as follows: In section \ref{hilb} we consider a natural graph Laplacian $\Delta_G$ on the infinite graph $G$ whose vertices form an $N$-fold branching tree. In Theorem \ref{propcycl}, we show that $\Delta_G$ has a natural representation in Voiculescu's Fock space $\F(H_N)$ where $H_N$ is an $N$-dimensional complex Hilbert space. The link between Laplacians on trees and the Fock space was first made in \cite{GG05}.

In Theorems \ref{thlapm}, \ref{thmucp}, and \ref{th3.25}, we give an explicit spectral representation of $\Delta_G$ with the spectral measure being Wigner's semicircle law (see \cite{Wig55}) or a related measure. Theorem \ref{th3.25} accounts for the spectral multiplicity of $\Delta_G$ as a selfadjoint operator. In section \ref{resi} we introduce the resistance metric for graphs in general; and (Proposition \ref{prop4.5}) we compute this metric explicitly for $\Delta_G$  corresponding to the case when $G$ is the graph of the $N$-fold branching tree. In Proposition \ref{proplat2}, this is contrasted with the analogous but much simpler case for the infinite lattice graphs.

 In section \ref{resi}, we establish the connection between the graph Laplacian $\Delta_G$ on one side, and on the other, a certain metric $d$ on the vertices $G^{(0)}$ of the graph, known in electrical models as the resistance metric. Using $\Delta_G$, we show that the $(G^{(0)}, d)$ embeds isometrically into a Hilbert space built directly from $\Delta_G$, the energy Hilbert space. In Proposition \ref{prop4.5}, we calculate this for the case of the tree-graphs. Proposition \ref{propsu1} makes the connection to random walk, return probability, and (Theorem \ref{thti1}) to path counting; computed from the moments of the spectral measure.

   In a general context graph Laplacians includes discrete Schrodinger operators, see e.g., \cite{ASM07}, and \cite{Sim95}. The present paper restricts its focus to the graphs $G$ having as vertex set trees built from the infinite iteration of an $N$-fold branching rule; but we contrast our results with those of other infinite graphs, for example lattice graphs. Moreover the study of $N$-fold branching rules is of independent interest, for example in signal processing, see e.g., \cite{Jor06b}. Another reason for the restriction in focus is that we then are able, with the use of Voiculescu's Fock space construction, to write down the complete spectral representation for the corresponding graph Laplacian, including a geometric model for the cyclic subspaces occurring in the spectral multiplicity table.

      While there is a number of related studies of graph Laplacians in the literature (e.g., \cite{JoPe08}), the detailed spectral picture has so far received relatively little attention.

        The focus of the present paper is instances of explicit spectral representations. While there is already a large literature on graph Laplacians, so far we have only encountered relatively few instances where the complete details are worked out for associated spectral representations. Since the geometric possibilities of graphs is vast, then so is the associated spectral configurations. Spectral approaches are manifold, and here our focus is: formulas for explicit spectral representations.
\subsection{Applications and connections to related results in the literature}
   A list of recent and past papers of relevance includes \cite{Str08}(received after we completed the present first draft); \cite{Car72, Car73a, Car73b, CR06, Chu07, CdV99, CdV04, Jor83}, and Wigner's original paper on the semicircle law \cite{Wig55}.

    While graph Laplacians have numerous applications (see our cited
references) in the theory of fractals \cite{Str08}, in combinatorics \cite{Chu07},
in random walk models \cite{SZ07},
 in free probability \cite{Voi85}, in operator theory \cite{HeHo73}, and in harmonic
analysis \cite{RPB07, Car73b, CdV99};
they are of relevance as well in mathematical physics; see e.g., \cite{ASW06,
AlFr00, Bre07, GG05, Moh91, RoRu95, Sol04}.
The mathematical physics literature contains a wealth of spectral
theoretical results, each with a particular focus.
Here we give formulas for the spectral transform of the tree-graph
Laplacians, including multiplicities.
And, conversely we show that the tree-graphs reflect themselves directly and
in an explicit manner in the spectral
 multiplicities for the associated operator.

 In Voiculescu's original approach to spectral representation, he relied on a
high powered result of Helton-Howe (traces, index and homology, \cite{HeHo73}).
While one does arrive at a representation this way this argument, and the
big machinery, leaves one wonder if there is instead a direct and
computational way of getting it. We find this!

While authors of earlier more ``practical'' uses of discrete Laplacians and
discrete Schroedinger operators have found
a variety of spectral features of related operators (see the cited
references above), none of them offer a complete spectral transform,
 and they do not have a spectral duality with Voiculescu's Toeplitz
$C^*$-algebra, as we do here. We feel that this is of independent interest.
It not only extends what was in the literature already, but it also offers a
unifying framework.

    The latter, plus our explicit formulas, is one of our key points in our
analysis and in our derivation of
complete spectral transforms. In fact we further give a {\it formula} for the
spectral transform.
And we use free probability methods from Voiculescu's Toeplitz $C^*$-algebra
\cite{Voi85} in pinning
down the complete spectral multiplicity picture for graph Laplacians on
trees. Conversely, we
use this in computing moments and resistance distances for the graphs
themselves.
Thus, our main result shows that the tree graphs reflect themselves directly
in the spectral multiplicities
for the associated graph Laplacian.
This may well work in both directions: Our spectral theoretic results may be
of use in free probability computations.

\section{Definitions}\label{defi}

\begin{definition}\label{def2.1}
Graph Laplacians. Let $G=(G^{(0)},G^{(1)})$ be a non-oriented graph with vertices $G^{(0)}$ and edges $G^{(1)}$. If $e\in G^{(1)}$ we assume that the source $s(e)$ and the terminal vertex $t(e)$ are different. If $(xy)\in G^{(1)}$ we write $x\sim y$. We further assume that for every $x\in G^{(0)}$ the set of neighbors
\begin{equation}\label{eq1.1}
\mbox{nbh}(x):=\{y\in G^{(0)}\,|\,y\sim x\}
\end{equation}
is {\it finite}.

Let $c:G^{(1)}\rightarrow\br_+$ be a fixed function (called {\it conductance}). Let $l^2(G^{(0)})$ be the Hilbert space of square summable sequences, i.e., 
\begin{equation}
	\|v\|_{l^2}^2:=\sum_{x\in G^{(0)}}|v(x)|^2<\infty,
	\label{eq1.2}
\end{equation}
and
\begin{equation}
	\ip{v_1}{v_2}_{l^2}:=\sum_{x\in G^{(0)}}\cj v_1(x)v_2(x),\quad\mbox{ for }v_1,v_2\in l^2(G^{(0)}).
	\label{eq1.3}
	\end{equation}

Let $\mathcal D$ be the set of all finitely supported elements in $l^2(G^{(0)})$, i.e., $v\in\mathcal D$ iff there exists $F\subset G^{(0)}$ ( a finite subset) such that $v(x)=0$ for all $x\in G^{(0)}\setminus F$.

Given $(G,c)$, the corresponding graph Laplacian $\Delta=\Delta_{G,c}$ is 
\begin{equation}
	(\Delta v)(x)=\sum_{y\sim x} c(xy)(v(x)-v(y)),\quad (v\in\mathcal D)
	\label{eq1.4}
\end{equation}
\end{definition}

It was proved in \cite{Jor08} that $\Delta$ is essentially selfadjoint, i.e., that the operator closure $\cj\Delta$ is selfadjoint. Hence for every $(G,c)$, there is a projection valued measure
\begin{equation}
	E_\Delta:\mathcal B([0,\infty))\rightarrow\mbox{Proj}(l^2(G^{(0)}))
	\label{eq1.5}
\end{equation}
where $\mathcal B$ stands for the Borel sigma-algebra, and $\mbox{Proj}$ is the lattice of all selfadjoint projections, i.e., $P=P^*=P^2$. The function $E_\Delta(\cdot)$ in \eqref{eq1.5} is countably additive on $\mathcal B$, and satisfies 
\begin{enumerate}
\item $$E_\Delta(A_1\cap A_2)=E_\Delta(A_1)E_\Delta(A_2),\quad(A_1,A_2\in\mathcal B);$$
\item $$I_{l^2}=\int_0^\infty E_\Delta(d\lambda);$$
\item $$\cj\Delta=\int_0^\infty\lambda E_\Delta(d\lambda).$$
\end{enumerate}
Moreover if $f$ is a Borel function, the operator $f(\Delta)$ is given by functional calculus,
$$f(\Delta)=\int f(\lambda)E_\Delta(d\lambda),$$ and a vector $v\in l^2(G^{(0)})$ is in the domain of $f(\Delta)$ iff
\begin{equation}
	\int_0^\infty|f(\lambda)|^2\|E_\Delta(d\lambda)v\|^2<\infty
	\label{eq1.6}
\end{equation}

For $v_0\in l^2(G^{(0)})$, $\|v_0\|=1$, set 
\begin{equation}\label{eq2.7}
\mu_0(\cdot):=\|E_\Delta(\cdot)v_0\|^2=\ip{v_0}{E_{\Delta}(\cdot)v_0}
\end{equation}.
We call $\mu_0$ the {\it spectral measure associated to $\Delta$ and the vector $v_0$}. See also Lemma \ref{lemr2} below.

For a measure $\mu_0$ on $\br$, define 
\begin{equation}
F_{\mu_0}(z):=\int_{\br}\frac{1}{x-z}\,d\mu_0(x),	
	\label{eqbor}
\end{equation}
the {\it Borel transform}, $z\in\bc\setminus\br$.

\begin{lemma}\label{lemr1}
If $R_\Delta(z):=(\Delta-z I_{l^2})^{-1}$ is the resolvent operator, then
$$F_{\mu_0}(z)=\ip{v_0}{R_\Delta(z)v_0},\quad z\in\bc\setminus\br.$$
\end{lemma}

\begin{proof}
Using (iii), we have
$$R_\Delta(z)=\int_\br(x-z)^{-1}E_\Delta(dx),$$
and for $z\in\bc\setminus\br$,
$$\ip{v_0}{R_\Delta(z)v_0}_{l^2}=\int_\br(x-z)^{-1}\ip{v_0}{E_\Delta(dx)v_0}=\int_{\br}(x-z)^{-1}\|E_\Delta(dx)v_0\|^2=\int_\br(x-z)^{-1}\,d\mu_0(x)=F_{\mu_0}(z).$$
\end{proof}

\begin{lemma}\label{lemr2}
Let $v_0$ and $\mu_0$ be as above, and let $\H_\Delta(v_0)\subset l^2(G^{(0)})$ be the $\Delta$-cyclic subspace generated by $v_0$. Set
$$\mathcal H_\Delta(v_0)\ni f(\Delta)v_0\stackrel{W}{\mapsto} f(\cdot)\in L^2(\mu_0).$$
Then $W$ extends to a unitary isomorphism of $\H_\Delta(v_0)$ into $L^2(\mu_0)$ which satisfies the intertwining relation $M_\lambda W=W\Delta$ on $\H_\Delta(v_0)$.
\end{lemma}
\begin{proof}
See e.g., \cite{Nel69,Sto32}.
\end{proof}

\begin{definition}\label{def1.2} 
Intertwining operators. 

Let $\H_i$, $i=1,2$ be Hilbert spaces, and let $W:\H_1\rightarrow \H_2$ be a unitary isomorphism of $\H_1$ into $\H_2$. Let $\mathcal F_i$, $i=1,2$ be families of operators in the respective Hilbert spaces. We say that $W$ is {\it intertwining}, or equivalently that the two families are {\it unitarily equivalent} if there is a bijection $\varphi:\mathcal F_1\rightarrow F_2$ such that
\begin{equation}
	WT=\varphi(T)W
	\label{eq1.7},\quad(T\in\mathcal F_1).
\end{equation} 

\end{definition}

\begin{definition}\label{def1.3}
Cuntz relations.

Let $H$ be a Hilbert space of finite dimension $N$. Let $\H$ be a second Hilbert space, infinite dimensional, and let $\mathcal B(\mathcal H)$ be the $C^*$-algebra of all bounded linear operators in $\H$. The norm in $\mathcal B(\H)$ will be the operator norm, i.e., if $T\in\mathcal B(\H)$
\begin{equation}
	\|T\|:=\sup\{\|Tv\|\,|\, v\in\mathcal H, \|v\|=1\}.
	\label{eq1.8}
\end{equation}

A linear function $\varphi:H\rightarrow\mathcal B(\H)$ with the property that
\begin{equation}
	\varphi(v)^*\varphi(v)=\|v\|^2 I_{\H},\quad(v\in H).
	\label{eqct}
\end{equation} 
 is called a {\it Cuntz-Toeplitz representation}, and the operators $\{\varphi(v)\,|\,v\in H\}$ are said to satisfy the Cuntz-Toeplitz relations. 
\end{definition}

\begin{lemma}\label{lem1.4}
Let $\varphi:H\rightarrow\mathcal B(\H)$ be as in Definition \ref{def1.3}. Then $\varphi$ is a Cuntz-Toeplitz respresentation if and only if for all orthonormal basis $|i\rangle$ in $H$, $i=1,\dots,N$, the operators
\begin{equation}
	T_i:=\varphi(|i\rangle)
	\label{eq1.9}
\end{equation}
satisfy
\begin{equation}
	T_i^*T_j=\delta_{i,j}I_{\H},\quad\mbox{ and }
\label{eq1.10i}
\end{equation}

\end{lemma}

\begin{proof}
Let $v\in H$, and set $v:=\sum_{i=1}^N v_i |i\rangle$. Then by linearity $\varphi(v)=\sum_{i=1}^Nv_iT_i$, see \eqref{eq1.9}. Moreover, using \eqref{eqct} and its polarization, we get
\begin{equation}
	\sum_{i=1}^N|v_i|^2=\|\varphi(v)\|^2=\|\varphi(v)^*\varphi(v)\|=\sum_{j=1}^N\sum_{k=1}^N\cj v_jv_kT_j^*T_k,\quad(v\in H).
	\label{eq1.11}
\end{equation}
The  relations in \eqref{eq1.10i} are directly equivalent to this.

\end{proof}

\begin{definition}\label{def1.5}
Let $N\in\bn$, $N\geq 2$, and let $\{T_i\}_{i=1}^N$ be a family of operators in a Hilbert space $\H$. We say that this is a representation of the {\it Cuntz algebra }$\mathcal O_N$ if \eqref{eq1.10i} holds and,
\begin{equation}
	\sum_{i=1}^NT_iT_i^*=I
	\label{eq1.12a}
\end{equation}
We say that this is a {\it Toeplitz system}, or a {\it representation of the Toeplitz algebra }$\mathcal T_N$ if \eqref{eq1.10i} holds, but
\begin{equation}
	\sum_{i=1}^NT_iT_i^*<I
	\label{eq1.12}
\end{equation}
\end{definition}

\begin{definition}\label{def1.6}
Unilateral shift.

Let $\H$ be a complex Hilbert space and let $T:\H\rightarrow \H$ be an isometric operator, i.e., an {\it isometry}. We say that $T$ is a {\it (unilateral) shift} if one (and hence all) of the following equivalent conditions are satisfied:
\begin{enumerate}
\item There exists a Hilbert space $H$ and a unitary isomorphism $W:\H\rightarrow{\sum_0^\infty}^\oplus H$ such that $W$ intertwines $T$ with 
\begin{equation}
	(x_0,x_1,x_2,\dots)\rightarrow(0,x_0,x_1,x_2,\dots)
	\label{eq1.13}
\end{equation}
on ${\sum_{0}^\infty}^\oplus H$. Then $\dim H$ is called the {\it multiplicity} of $T$.

\item For any $x\in\mathcal H$, $\lim_{n\rightarrow\infty}{T^*}^nx=0$.
\item  The projections $P_n:=T^n{T^*}^n$ satisfy $P_1\geq P_2\geq\dots\geq P_n\geq P_{n+1}\geq\dots$ and $\inf_nP_n=0$.
\end{enumerate}
See \cite{NaFo70}.
\end{definition}

\begin{definition}\label{def1.7}
Bilateral shift.

Let $\H$ be a complex Hilbert space and let $T:\H\rightarrow\H$ be an isometry. We say that $T$ is a bilateral shift if one (and hence all) of the following conditions is satisfied:
\begin{enumerate}
\item There exists a Hilbert space $H$ and a unitary isomorphism $W:\mathcal H\rightarrow{\sum_{-\infty}^\infty}^\oplus H$ such that $W$ intertwines $T$ with 
$$(x_n)_{n\in\bz}\rightarrow (x_{n-1})_{n\in\bz}$$
on ${\sum_{-\infty}^\infty}^\oplus H$.
\item
There is a unitary isomorphism $W:\H\rightarrow L^2(\bt,H)$ such that $W$ intertwines $T$ with the multiplication operator
$M:f\rightarrow zf(z)$ on $L^2(\bt,H)$. Here $L^2(\bt,H)$ denotes the Hilbert space of vector valued functions $\bt\rightarrow H$, measurable and satisfying $\int_{\bt}\|f(z)\|^2\,d\nu(z)<\infty$, where $\nu$ is the Haar measure on circle $\bt:=\{z\in\bc\,|\,|z|=1\}$.
\end{enumerate}
\end{definition}

\section{Hilbert spaces and a graph model}\label{hilb}

We begin with a particular graph, called the tree with $N$-fold branching. It is formed out of an alphabet, say $A$ of size $N$ as follows. We first form words $W_k = W_k(A)$ of length $k$, for each $k$, with letters chosen from $A$. It is often convenient to take $A$ to be the cyclic group $\bz_N$ of order $N$. The tree $V = V_N$  will be the union of the sets $W_k$ for $k = 0, 1,\dots,$ such that the $k=0$ case corresponds to the empty word. We then form a graph $G_N$ with vertices $V = V_N$ consisting of the points in the tree with $N$-fold branching. Specifically $V$ is the union of the words $W_k$, with the understanding the $W_0$ is the empty word. The edges in $G_N$ will be made up of nearest neighbors in $V$ as follows: If $x$ is in $W_0$, the set of edges emanating at $x$ consists simply of the points in $A$. If $x$ is in $W_k$ , $k\geq 1$, then the edges in $G_N$ consist of the lines connecting two vertices, $x$ one of them, and the others resulting from $x = (x_1x_2\dots x_k)$ by truncation at the tail end and by addition of a letter from $A$ also in the tail end. So the nearest neighbors for the vertex $x$ consist of the $N + 1$ vertices  $(x_1x_2 \dots  x_{k-1})$ and $(xa)$ where $a$ is chosen from $A$; so $N+1$ nearest neighbors in all.

      Starting with this particular tree graph $G_N$ there are two naturally associated Hilbert spaces, the first is simply the $l^2$-space over $V_N$, and the second is a Fock space $\mathcal F(H) = \mathcal F(H_N)$, $H_N := l^2(\bz_N)=\bc^N$ which carries additional structure. This additional features of the Hilbert space $\mathcal F(H)$ will be needed later.

      In Proposition \ref{proplapf} we show that the two Hilbert spaces are naturally isomorphic. While this is known, we have included a proof sketch in order to fix our notation. The Fock space $\mathcal F(H)$ is used in particle physics in the conventional description of (infinite) quantum systems with particles governed by Boltzmann statistics.

\begin{definition}\label{def3.1}
The infinite $N$-ary tree is the graph with vertices 
$$V:=\{\emptyset\}\cup\left\{\omega_1\dots\omega_n\,|\,n\geq 1, \omega_1,\dots,\omega_n\in\{1,\dots,N\}\right\},$$
and edges given by the relations $\emptyset\sim i$, and $\omega_1\dots\omega_n\sim \omega_1\dots\omega_n i$ for all $n\geq 1$, $i,\omega_1\dots\omega_n\in\{1,\dots,N\}$.
\par
Thus $V$ is the set of finite words over the alphabet $\{1,\dots,N\}$, including the empty word $\emptyset$, and we have edges between a word $\omega$ and the word $\omega i$ obtained from $\omega$ by adjoining a letter at the end. 
\end{definition}

\begin{figure}[h]
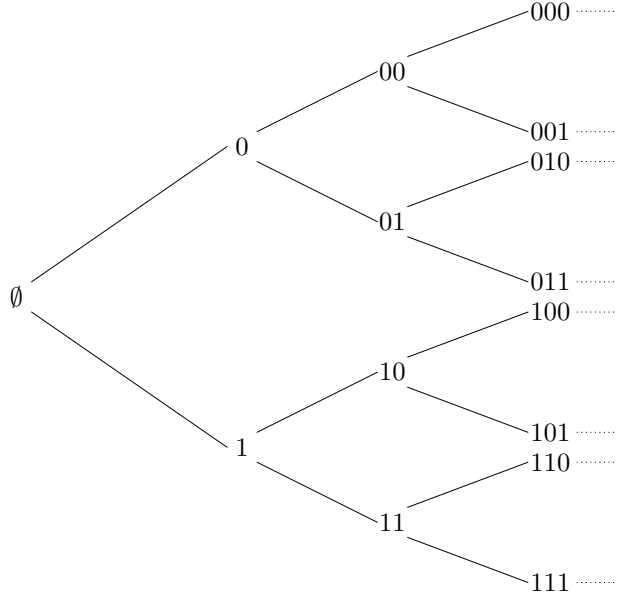

\caption{The binary tree $N=2$.}\label{fig1}
\[
\xy
(-78,2); (-52,20)**@{-};
(-80,0)*{\ty};
(-78,-2);(-52,-20)**@{-};
(-50,20)*{0};
(-50,-20)*{1};
(-48,22);(-32,30)**@{-};
(-48,18);(-32,10)**@{-};
(-48,-18);(-32,-10)**@{-};
(-48,-22);(-32,-30)**@{-};
(-30,30)*{00};
(-30,10)*{01};
(-30,-10)*{10};
(-30,-30)*{11};
(-28,32);(-12,38)**@{-};
(-28,28);(-12,22)**@{-};
(-28,12);(-12,18)**@{-};
(-28,8);(-12,2)**@{-};
(-28,-8);(-12,-2)**@{-};
(-28,-12);(-12,-18)**@{-};
(-28,-28);(-12,-22)**@{-};
(-28,-32);(-12,-38)**@{-};
(-9,38)*{000};
(-9,22)*{001};
(-9,18)*{010};
(-9,2)*{011};
(-9,-38)*{111};
(-9,-22)*{110};
(-9,-18)*{101};
(-9,-2)*{100};
(-6,38);(0,38)**@{.};
(-6,22);(0,22)**@{.};
(-6,18);(0,18)**@{.};
(-6,2);(0,2)**@{.};
(-6,-2);(0,-2)**@{.};
(-6,-18);(0,-18)**@{.};
(-6,-22);(0,-22)**@{.};
(-6,-38);(0,-38)**@{.};
\endxy
\]
\end{figure}

\begin{definition}
Define the maps $\sigma$, $\tau_i$, $i\in\{1,\dots,N\}$ on $V$:
\begin{equation}\label{eqsigma}
\sigma(\emptyset)=\emptyset,\quad \sigma(\omega_1\dots\omega_n)=\omega_1\dots\omega_{n-1},\quad(n\geq 1,\omega_1,\dots\omega_n\in \{1,\dots,N\})
\end{equation}
Thus $\sigma(i)=\emptyset$ for all $i\in\{1,\dots,N\}$, so $\sigma$ removes the last letter of the word.

\begin{equation}\label{eqtau}
\tau_i(\emptyset)=i,\quad \tau_i(\omega_1\dots\omega_n)=\omega_1\dots\omega_n i,\quad(\omega_1,\dots,\omega_n\in\{1,\dots,N\}, i\in\{1,\dots,N\})
\end{equation}
\end{definition}
So $\tau_i$ adjoins the letter $i$ at the end of the word.
\par
Define the operators $U$ and $S_i$, $(i\in\{1,\dots,N\})$ on $l^2(V)$ by
\begin{equation}\label{equ}
Uv=v\circ\sigma,\quad S_iv=v\circ\tau_i,\quad(i\in\{1,\dots,N\})
\end{equation}

In the following we show that the operators $S_i$ the $l^2(V)$ as an operator
system generate a representation of the Toeplitz algebra in several
variables (see \cite{Voi85}.) In \cite{Voi85} this Toeplitz algebra is realized in the
Fock space $\mathcal F(H)$. The Hilbert space $\mathcal F(H)$ carries a system of creation
operators $T_v$ indexed by vectors $v$ in $H$.

 As shown in formula \eqref{eqtv*}, the corresponding adjoint operators $T_v^*$ are
annihilation operators. In analyzing these operators, Dirac's bra-ket
notation is convenient. We will select an orthonormal basis (ONB) $\{|i\rangle\}$ in
$H$, and denote the corresponding creation operators $T_i$ , i.e, $T_i = T_{|i\rangle}$.

  In  Proposition \ref{proplapf}, we show that the unitary isomorphism $W : l^2(V)\rightarrow \mathcal F(H)$ intertwines the operators $S_i$ with the corresponding system $T_i^*$ in $\mathcal F(H)$.

   This helps us identify a universal Hilbert space representation for
iterated function systems(IFSs), and to make precise the quotient of the
Toeplitz algebra by the compact operators as a Cuntz algebra (see \cite{Jor04}).

The next two lemmas require some simple computations.

\begin{lemma}\label{lems_i}
For $f\in l^2(V)$, $i\in\{1,\dots,N\}$ one has:
\begin{enumerate}
\item For all $\omega_1,\dots,\omega_n\in\{1,\dots,N\}$:
$$S_i\delta_\ty=0,\quad S_i\delta_{\omega_1\dots\omega_n}=\left\{\begin{array}{cc}\delta_{\omega_1\dots\omega_{n-1}},&\mbox{ if }\omega_n=i\\
0,&\mbox{ if }\omega_n\neq i\end{array}\right.=\delta_{\omega_n,i}\delta_{\omega_1\dots\omega_{n-1}}.$$
\item For all $\omega_1,\dots,\omega_n\in\{1,\dots,N\}$:
$$(S_i^*f)(\ty)=0,\quad (S_i^*f)(\omega_1\dots\omega_n)=\left\{\begin{array}{cc} f(\omega_1\dots\omega_{n-1}),&\mbox{ if }\omega_n=i\\
0,&\mbox{ if }\omega_n\neq i\end{array}\right.=\delta_{\omega_n,i}f(\omega_1\dots\omega_{n-1}).$$
\item For all $\omega_1,\dots,\omega_n\in\{1,\dots,N\}$:
$$S_i^*\delta_\ty=\delta_i,\quad S_i^*\delta_{\omega_1\dots\omega_n}=\delta_{\omega_1\dots\omega_ni}.$$
\end{enumerate}
\end{lemma}

\begin{lemma}\label{lemtoepls_i}
The operators $S_i$ satisfy the following relations:
\begin{enumerate}
	\item $S_iS_i^*=I$ for all $i\in\{1,\dots,N\}$.
	\item Let $P_\ty$ be the projection in $l^2(V)$ onto the vector $\delta_\ty$. 
	$$\sum_{i=1}^NS_i^*S_i=I-P_\ty.$$
\end{enumerate}
\end{lemma}

\begin{proposition}\label{propls}
Let $P_\emptyset$ be the projection in $l^2(V)$ onto the canonical vector $\delta_\emptyset$.
\begin{enumerate}
\item $U^*=\sum_{i=1}^NS_i+P_\ty$;
\item
$$\Delta=(N+1)I-(\sum_{i=1}^NS_i+U)=(N+1)I-(U+U^*-P_\ty)=(N+1)I-(\sum_{i=1}^NS_i+\sum_{i=1}^NS_i^*+P_\ty).$$
\end{enumerate}
\end{proposition}

\begin{proof}
(i) For $v,v'$ in $l^2(V)$, we have
$$\ip{Uv}{v'}=\cj v(\ty)v'(\ty)+\sum_{n\geq1, \omega_1,\dots,\omega_n}\cj v(\omega_1\dots\omega_{n-1})v'(\omega_1\dots\omega_n)=$$
$$\cj v(\ty)v'(\ty)+\cj v(\ty)\sum_{\omega_1}v'(\omega_1)+\sum_{\omega_1,\dots,\omega_{n-1}}\cj v(\omega_1\dots\omega_{n-1})\sum_{\omega_n}v(\omega_1\dots\omega_{n-1}\omega_n).$$
Therefore
$$(U^*v)(\ty)=v(\ty)+\sum_{i=1}^Nv(i),\quad (U^*v)(\omega_1\dots\omega_n)=\sum_{i=1}^Nv(\omega_1\dots\omega_ni).$$
This implies that $U^*=\sum_{i=1}^N S_i+ P_\ty$.

\par
(ii) We have 
$$(\Delta v)(\ty)=Nv(\ty)-\sum_{i=1}^Nv(i)=(N+1)v(\ty)-(\sum_{i=1}^Nv(i)+v(\ty)),$$
and
$$(\Delta v)(\omega_1\dots\omega_n)=(N+1)v(\omega_1\dots\omega_{n})-(v(\omega_1\dots\omega_{n-1})+\sum_{i=1}^Nv(\omega_1\dots\omega_ni)).$$
This implies that $\Delta=(N+1)I-(\sum_{i=1}^NS_i+U)$. The other equalities follow from (i).
\end{proof}

\begin{remark}\label{remlat}
  In equation \eqref{eq1.4} we introduced our general class of graph Laplacians $\Delta_{G,c}$. As noted there, there is a Laplace operator for each graph $G$ and for each choice of conductance function $c$. Our paper is about spectral theory of the graph Laplacians, and the spectrum depends on both the graph, and the choice of conductance function.

   To understand the graph dependence, we may contrast two cases: Case 1: the $N$-fold tree graphs (Definition \ref{def3.1} and Figure \ref{fig1}) and the use of the unilateral shift (Definition \ref{def1.6}). Case 2: Lattice graphs (details below). For the latter we show that the spectrum is determined by the bilateral shift (Definition \ref{def1.7}.)

Let $G=(G^{(0)},G^{(1)})$ be the rank $d$ group with vertices $G^{(0)}=\bz^d$, and edges defined by $n=(n_1,\dots,n_d)\sim m=(m_1,\dots,n_d)$ iff there exists $k\in\{1,2,\dots,d\}$ such that $|m_k-n_k|=1$ and $m_j=n_j$ for $j\neq k$.

For $v\in l^2(\bz^d)$ set 
\begin{equation}
	(\Delta v)(n):=\sum_{m\sim n}v(n)-v(m);
	\label{eqlat1}
\end{equation}
for example, if $d=1$,
$$(\Delta v)(n)=2v(n)-v(n-1)-v(n+1).$$
These operators are used in numerical analysis, in electrical network analysis, in electrical network models in physics; (see e.g., \cite{Kig03,Pow76}) but with a variety of choices of the conductance function; see \eqref{eq1.4}

\begin{proposition}\label{proplat2}
Let $d\in\bn$. The graph Laplacian $\Delta$ in \eqref{eqlat1} has the form $I-2\preal T_d$ where $\preal T_d=\frac{1}{2}(T_d+T_d^*)$, $T_d=\otimes_{1}^d T$ and $T$ is the bilateral shift; see Definition \ref{def1.7}.
\end{proposition}
\begin{proof}
Introducing the Fourier transform on $\bt^d$ (the $d$-torus), we get the unitary equivalence
\begin{equation}
	v(z)=\sum_{n\in\bz^d}v_nz^n\leftrightarrow(v_n)_n\in l^2(\bz^d)
	\label{eqlat2}
\end{equation}
and by Parseval's equality
$$\|v\|_{L^2(\bt^d)}^2=\sum_{n\in\bz^d}|v_n|^2.$$
Substituting \eqref{eqlat2} into \eqref{eqlat1}, we get
$$(\Delta v)(z)=(2N-\sum_{k=1}^d z_k-\sum_{k=1}^d\cj z_k)v(z)=(2N-\sum_{k=1}^d 2\preal z_k)v(z).$$
Or setting $z_k=e^{ix_k}$, $k=1,\dots, d$, 
$$(\Delta v)(x)=2\sum_{k=1}^d(1-\cos x_k)v(x)=4\sum_{k=1}^d\sin^2\left(\frac{x_k}{2}\right)v(x).$$
\end{proof} 
\end{remark}

\begin{remark}\label{remlap2}
There are several contrasts between this case and the Laplace operators of tree graphs, see Theorem \ref{th3.25}: 
One is that the spectral measure here simply is a pull-back of the Haar measure on $\bt^d$ with the function given above; second that there is no rank one perturbation; and third that the spectrum is simple in the present bilateral case.

We also note that the Laplace operator for the lattice case has simple spectrum while the case of the tree has an intricate multiplicity structure, see Theorem \ref{propcycl} for specifics. 
\end{remark}
\subsection{The Fock space}

\begin{definition}\label{deffock}
Let $H:=\bc^N$ and $\Omega$ be a fixed unit vector. Then the Fock space is the Hilbert space 
$$\mathcal F:=\F(H):=\bc\Omega\oplus H\oplus (H\otimes H)\oplus\dots\oplus \underbrace{H\otimes\dots \otimes H}_{n\mbox{ times}}\dots.$$
For a vector $v\in H$ we define the operator $T_v$ on $\mathcal F$ by
\begin{equation}\label{eqtv}
T_vx:=x\otimes v,\quad x\in \mathcal F, (T_v\Omega:=v).
\end{equation}
\end{definition}

\begin{proposition}\label{proptv*}
For $v\in H$, the adjoint of the operator $T_v$ is given by the formula
\begin{equation}
	T_v^*\Omega=0,\quad T_v^*(x_1\otimes\dots\otimes x_n)=x_1\otimes\dots\otimes x_{n-1}\ip{v}{x_n},\quad (n\geq1, x_1,\dots, x_n\in H),\quad T_v^*\Omega=0.
	\label{eqtv*}
\end{equation}
\end{proposition}

The following Proposition summarizes a number of geometric properties of our isometries in the Fock space. We begin with the operators  $T$ and $T^*$ for $T=T_v$. In Definition \ref{def1.5} we emphasize that a unilateral shift is really an isomorphism class; referring to unitary equivalence, and that the multiplicity is a complete isomorphism invariant. Specifically, in the Proposition below, we spell out a particular shift representation for the operator $T_v$ when $\| v \| = 1$; and this refers to the Fock space which we will need later. In fact, for a fixed $v$, we must identify the closed subspace in the Fock space which is shifted by powers of the isometry $T_v$.

  We use Voiculescu's framework \cite{Voi85} and \cite{Ora01}. Moreover we include here the details we will be using later in the proof of our main conclusions regarding spectral representations of graph Laplacians;-  in the case of trees, these are operators derived from the $T_v$ system.

\begin{proposition}\label{propn1}
Let $N\in\bn$, and let $H$ be a Hilbert space of dimension $N$. Let $\mathcal F:=\mathcal F(H)$ be the Fock space (see Definition \ref{deffock}.) Let $v\in H$ satisfy $\|v\|=1$. Then $T=T_v$ is a shift operator in $\mathcal F$. The multiplicity of $T$ is as follows: $\operatorname*{mult}(T)=1$ iff $N=1$. And if $N>1$, then $\operatorname*{mult}(T)=\infty$.
\end{proposition}

\begin{proof}
It follows from Proposition \ref{proptv*} that $T$ is isometric, $T:\mathcal F\rightarrow\mathcal F$, and that
$$\Omega\in N(T^*)=\{\xi\in\mathcal F\,|\, T^*\xi=0\}.$$
If $N=1$, then $N(T^*)=\bc\Omega$; and otherwise $N(T^*)$ contains in addition the following vectors in $\mathcal F$: $\xi_1\otimes\dots\otimes\xi_{k-1}\otimes\xi_k$ where $\xi_i\in H$, and $\xi_k\perp v$, i.e., $\ip{\xi_k}{v}=0$.

Moreover, together with $\Omega$, the closed span of these vectors exhausts $N(T^*)$.

It is true in general for a fixed isometry $(T,\mathcal H)$ that the direct sum Hilbert space
\begin{equation}
	{\sum_{k=0}^\infty}^{\oplus}N(T^*)
	\label{eqn1}
\end{equation}
is isomorphic to a closed subspace in $\H$. Specifically, if $x_k\in N(T^*)$, then
\begin{equation}
	\sum_{k=0}^\infty T^kx_k\rightarrow x_0\oplus x_1\oplus x_2\oplus\dots
	\label{eqn2}
\end{equation}
is a well defined isometry, i.e., it satisfies 
\begin{equation}
	\|\sum_{k=0}^\infty T^kx_k\|_{\H}^2=\sum_{k=0}^\infty\|x_k\|^2=\|{\sum_{k=0}^\infty}^\oplus x_k\|^2<\infty
	\label{eqn3}
\end{equation}

This is a consequence of the following identities:
\begin{equation}
	\ip{T^jx_j}{T^kx_k}_{\H}=\delta_{j,k}\ip{x_j}{x_k},\quad(j,k\in\bn_0, x_j\in N(T^*)).
	\label{eqn4}
\end{equation}
Introducing $P_k:=T^k{T^*}^k$ from \eqref{eqn2}-\eqref{eqn3}, we see that a vector $\xi\in\mathcal H$ satisfies $\inf_k\|P_k\xi\|^2=0$ iff $\xi\in \vee_k\operatorname*{Ran}(T^k)=(\wedge_k N({T^*}^k))^{\perp}$.

Set $\H=\mathcal F=\mathcal F(H)$, and assume $N>1$. Set $T=T_v$, for $v\in H$, $\|v\|=1$. Then 

\begin{equation}
	P_{k+1}H^{\otimes k}=0,\quad (k\in\bn_0)
	\label{eqn5}
\end{equation}
with the convention $H^{\otimes 0}=\bc\Omega$.

If $k<m$, then 
\begin{equation}
	P_k|_{\otimes_0^m H}=(I_{\otimes_0^{m-k}})\otimes |\underbrace{v\otimes\dots\otimes v}_{k\mbox{ times }}\rangle\langle\underbrace{v\otimes\dots\otimes v}_{k\mbox{ times }}|
	\label{eqn6}
\end{equation}
where we use Dirac's notation $|\cdot\rangle\langle\cdot|$ is denoting rank one operators.

Since $\mathcal F={\sum_{k\geq0}}^\oplus H^{\times k}$, we have the following representation:
$$\mathcal F\ni\xi=\sum_{k\geq0}\xi_k,\quad\xi_k\in H^{\otimes k};$$
and
\begin{equation}
	\|\xi\|_{\mathcal F}^2=\sum_{m\geq 0}\|\xi_m\|^2_{\otimes_0^m H}
	\label{eqn7}
\end{equation}
Using \eqref{eqn5}-\eqref{eqn6}, we then get
$$\|P_k\xi\|^2=\sum_{m\geq k}\|P_k\xi_m\|^2\leq \sum_{m\geq k}\|\xi_m\|^2\rightarrow 0\mbox{ as } k\rightarrow\infty;$$
where we used \eqref{eqn3} in the last step.

Hence $T$ is a unilateral shift according to Definition \ref{def1.6}(iii).

It follows from \eqref{eqn6} that $\operatorname*{mult}(T)=1$ iff $N=\dim H=1$.

If $N>1$, then $N(T^*)$ is the closure of the span of $\Omega$ and the infinite sequence of closed subspaces 
$$(\otimes_0^kH)\otimes(H\ominus\{v\}).$$
Hence this sum is an infinite dimensional Hilbert space. We used the terminology:
$$H\ominus\{v\}:=\{w\in H\,|\,\ip{w}{v}=0\}.$$

\end{proof}
\begin{definition}\label{defn2}
Let $H$ be an $N$-dimensional Hilbert space and let $\mathcal F=\mathcal F(H)$ be the Fock space. Let $v\in H$, $\|v\|=1$ be given. Set $\operatorname*{Ran}(T_v)=$the range of the isometry $T_v$. Then $T_vT_v^*$ is the projection onto $\operatorname*{Ran}(T_v)$, and $I-T_vT_v^*$ is the projection onto $N(T_v^*)$. 
\end{definition}

\begin{corollary}\label{corn3}
Let $H$ and $\mathcal F=\mathcal F(H)$ be as in Proposition \ref{propn1}. Suppose $N=\dim H>1$, and let $v\in H$, $\|v\|=1$ be given. Then 
\begin{equation}
	\bc\Omega\oplus\sum_{w\in H\ominus \{v\}}\operatorname*{Ran}(T_w)=N(T_v^*)
	\label{eqn8}
\end{equation}
\end{corollary}
\begin{proof}
Set $v_1=v$, and extend to an ONB: $v_2,\dots,v_N$ for $H$. Then we have
\begin{equation}
	I_{\mathcal F}-T_{v_1}T_{v_1}^*=P_{\Omega}+\sum_{i=2}^NT_{v_i}T_{v_i}^*
	\label{eqn9}
\end{equation}
But we noted that $I_{\mathcal F}-T_{v_1}T_{v_1}^*$ is the projection onto $N(T_{v_1})$, while the right-hand side in \eqref{eqn9} is the projection onto 
$$\bc\Omega\oplus\sum_{w\perp v}\operatorname*{Ran}(T_w).$$
\end{proof}

\begin{proposition}\label{proplapf} 
Consider the infinite $N$-ary tree in Definition \ref{def3.1}.
Let $(e_i)_{i=1}^N$ be an orthonormal basis for $H=\bc^N$. And let $\F=\F(H)$ be the corresponding Fock space. The map $\mathcal W: l^2(V)\rightarrow\mathcal F$ defined by
\begin{equation}
	\mathcal W(\delta_\ty)=\Omega,\quad \mathcal W(\delta_{\omega_1\dots\omega_n})=e_{\omega_1}\otimes\dots\otimes e_{\omega_n},\quad(\omega_1,\dots,\omega_n\in \{1,\dots,N\})
	\label{eqW}
\end{equation}
defines an isometric isomorphism between the two Hilbert spaces. 
\par
Let $s_\Delta:=\sum_{i=1}^N e_i$. Then 

\begin{equation}\label{eqsf}
\mathcal WS_i\mathcal W^*=T_{e_i}^*,\quad(i\in\{1,\dots,N\}),
\end{equation}
and
\begin{equation}
	\mathcal W\Delta\mathcal W^*=(N+1)I-(T_{s_\Delta}+T_{s_{\Delta}}^*+P_{\Omega})
	\label{eqlapf}
\end{equation}
where $P_\Omega$ is the projection in $\mathcal F$ onto the vector $\Omega$.
\end{proposition}

\begin{proof}
The map $\mathcal W$ defines an isometric isomorphism because it maps the orthonormal basis $(\delta_\omega)_{\omega\in\Omega}$ into an orthonormal basis $\{\Omega\}\cup\{e_{\omega_1}\otimes\dots\otimes e_{\omega_n}\,|\,\omega_1,\dots,\omega_n\in\{1,\dots,N\}, n\in\bn\}$.

Next, we check equation \eqref{eqsf} on this orthonormal basis.
$$\mathcal WS_i\mathcal W^*\Omega=\mathcal WS_i\delta_\ty=\mathcal W\delta_\ty\circ\tau_i=0=T_{e_i}^*\Omega.$$
$$\mathcal WS_i\mathcal W^*e_{\omega_1}\otimes\dots\otimes e_{\omega_n}=WS_i\delta_{\omega_1\dots\omega_n}=W\delta_{\omega_1\dots\omega_n}\circ\tau_i=
W\delta_{\omega_1\dots\omega_{n-1}}\delta_{\omega_n,i}=$$$$e_{\omega_1}\otimes\dots\otimes e_{\omega_{n-1}}\ip{e_{\omega_n}}{e_i}=T_{e_i}^*e_{\omega_1}\otimes\dots\otimes e_{\omega_n}.$$
This implies \eqref{eqsf}

It is clear that $\mathcal WP_\ty\mathcal W^*=P_\Omega$.

From equation \eqref{eqsf} and Proposition \ref{propls} it follows that 
$$\mathcal W\Delta\mathcal W^*=(N+1)I-(\sum_{i=1}^NT_{e_i}^*+\sum_{i=1}^NT_{e_i}+P_\Omega).$$
But since $\sum_{i=1}^N T_{e_i}=T_{s_\Delta}$, equation \eqref{eqlapf} follows.
\end{proof}

\def\H{\mathcal H}

\subsection{Decomposition into cyclic subspaces}
\begin{definition}\label{defh_i}
Let $\{x_1,\dots,x_N\}$ be an orthonormal basis for $H=\bc^N$ with the first vector $x_1=s_0:=\frac{s_\Delta}{\|s_\Delta\|}=\frac{1}{\sqrt{N}}s_\Delta$. ($s_\Delta$ is defined in Proposition \ref{proplapf}).
Define the following subspaces of $\mathcal F$:
$$\H_{\Omega}:=\cj{\mbox{span}}\{\Omega,\otimes_{k=1}^p s_0\,|\,p\geq 1\},$$ 
$$\H_{i_1\dots i_n}:=\cj{\mbox{span}}\{x_{i_1}\otimes\dots \otimes x_{i_n}\otimes\otimes_{k=1}^ps_0\,|\,p\geq 0\},\quad i_1,\dots,i_n\in\{1,\dots, N\}, i_n\neq 1, n\geq 1.$$
so $x_{i_n}$ should not be the vector in the basis corresponding to $s_\Delta$.
\end{definition}

\begin{theorem}\label{propcycl}
The cyclic subspace decomposition of the graph Laplacian.
\begin{enumerate}
\item
The subspaces $\H_{\Omega}$ and $\H_{i_1\dots i_n}$, $i_1,\dots, i_n\in\{1,\dots,N\}$, $i_n\neq 1$ are mutually orthogonal, 

$$\H_\Omega\oplus\oplus_{i_1,\dots, i_n\in\{1,\dots,N\}, i_n\neq 1} \H_{i_1\dots i_n}=\mathcal F,	$$

and they are cyclic subspaces for the operator $\mathcal W\Delta\mathcal W^*$.
\item
Let $S$ be the unilateral shift of multiplicity one, 
$$S:l^2(\bn_0)\rightarrow l^2(\bn_0), S(x_0,x_1,x_2,\dots)=(0, x_0,x_1,x_2,\dots).	$$

The restriction of $\mathcal W\Delta\mathcal W^*$ to $\H_\Omega$ is unitarily equivalent to 
the operator 
\begin{equation}
D_{\Omega}:=	(N+1)I_{l^2(\bn_0)}-2\sqrt{N}{\operatorname*{Re} S}-P_{\delta_0}.
	\label{eqlapsh1}
\end{equation}
For all $i_1,\dots,i_n\in\{1,\dots,N\}$, $i_n\neq 1$, and $n\geq1$, the restriction of $\mathcal W\Delta\mathcal W^*$ to $\H_{i_1\dots i_n}$ is unitarily equivalent to the operator 
\begin{equation}
D:=(N+1)I_{l_2(\bn)}-2\sqrt{N}{\operatorname*{Re} S}.	
	\label{eqlapsh2}
\end{equation}

Where $\operatorname*{Re} S=\frac{S+S^*}2$, and $P_{\delta_0}$ is the projection in $l^2(\bn_0)$ onto the vector $\delta_0$.
\end{enumerate}
\end{theorem}

\begin{proof} (i)
To prove that the spaces are mutually orthogonal, take $x_{i_1}\otimes\dots\otimes x_{i_n}\otimes\otimes_{k=1}^p s_0$ and 
$x_{j_1}\otimes\dots\otimes x_{j_m}\otimes\otimes_{k=1}^qs_0$. If the lengths are different, i.e. $n+p\neq m+q$, then the two vectors are orthogonal. If $n+p=m+q$ and $n<m$, then since $j_m\neq 1$, it follows that $s_0\perp x_{j_m}$, and again the vectors are orthogonal. Similarly, if $n>m$. If, in addition $n=m$, then if $i_1\dots i_n\neq j_1\dots j_m$ then $i_k\neq j_k$ and $x_{i_k}\perp x_{j_k}$ and the two vectors are orthogonal.

To prove that these subspaces span the entire space, we have that $\{\Omega, x_{i_1}\otimes\dots\otimes x_{i_n}\,|\, i_k\in\{1,\dots,N\}\}$ is an orthonormal basis for $\mathcal F$. If all $i_k=1$, then $x_{i_1}\otimes\dots\otimes x_{i_n}$ is in $\H_\Omega$. If some $i_k\neq 1$, take the last one as such, and $x_{i_1}\otimes\dots\otimes x_{i_n}\in\H_{i_1\dots i_k}$.

It remains to prove that these subspaces are cyclic for $\mathcal W\Delta\mathcal W^*$, or equivalently by Proposition \ref{proplapf}, for $A:=T_{s_\Delta}+T_{s_\Delta}^*+P_\Omega$. 
 Note that $T_{s_\Delta}=\sqrt{N}T_{s_0}$, so $A=\sqrt{N}(T_{s_0}+T_{s_0}^*)+P_\Omega$.

 Take $n\geq 1$. Take $x_{i_1}\otimes\dots\otimes x_{i_n}$ with $i_1,\dots ,i_n\in\{1,\dots,N\}$ and $i_n\neq1$. Then 
 $$A(x_{i_1}\otimes\dots\otimes x_{i_n})=\sqrt{N}x_{i_1}\otimes\dots\otimes x_{i_n}\otimes s_0+0+0.$$
 
 Assume by induction on $p$ that $x_{i_1}\otimes\dots\otimes x_{i_n}\otimes\otimes_{k=1}^ps_0$ is in the cyclic subspace of $A$ generated by the vector $x_{i_1}\otimes\dots\otimes x_{i_n}$. Then 
 $$Ax_{i_1}\otimes\dots\otimes x_{i_n}\otimes\otimes_{k=1}^ps_0=\sqrt{N}x_{i_1}\otimes\dots\otimes x_{i_n}\otimes\otimes_{k=1}^{p+1}s_0+\sqrt{N}x_{i_1}\otimes\dots\otimes x_{i_n}\otimes\otimes_{k=1}^{p-1}s_0+0,$$
 so $x_{i_1}\otimes\dots\otimes x_{i_n}\otimes\otimes_{k=1}^{p+1}s_0$ is in the same cyclic subspace. This implies that $\H_{i_1\dots i_n}$ is the cyclic subspace of the operator $A$ (hence of $\mathcal W\Delta\mathcal W^*$), generated by the vector $x_{i_1}\otimes\dots\otimes x_{i_n}$.
 
 Similarly for $\H_{\Omega}$.

(ii) For $i_1\dots i_n\in\{1,\dots,N\}$, $i_n\neq 1$, define the operator $\mathcal W_{i_1\dots i_n}:l^2(\bn_0)\rightarrow\H_{i_1\dots i_n}$,
$$\mathcal W_{i_1\dots i_n}\delta_p=x_1\otimes\dots\otimes x_{i_n}\otimes\otimes_{k=1}^ps_0,\quad (p\geq0).$$
The operator $ \mathcal W_{i_1\dots i_n}$ is clearly unitary since it maps an ONB to an ONB.
A simple calculation shows that $\W_{i_1\dots i_n}$ intertwines $D$ and $\mathcal W^*\Delta\W$ on $\H_{i_1\dots i_n}$. Similarly for $\H_\Omega$.
 
\end{proof}

\begin{lemma}\cite{Voi85}\label{lemvoic}
 The spectral measure of $\operatorname*{Re} S$.
Let $S$ be the unilateral shift of multiplicity one on $l^2(\bn_0)$. Let $\mu_c$ be Wigner's semicircular measure on $[-1,1]$, i.e. $\mu_c$ is absolutely continuous with respect to the Lebesgue measure on $[-1,1]$, and  
$$\frac{d\mu_c}{dx}=\left\{\begin{array}{cc}\frac{2}{\pi}\sqrt{1-x^2},&\mbox{ if }x\in[-1,1]\\
0,&\mbox{ otherwise.}\end{array}\right.$$

There exists an isometric isomorphism $\Phi:l^2(\bn_0)\rightarrow L^2(\mu_c)$, such that $\Phi(\delta_0)=1$ (the constant function $1$), and 
$$\Phi\preal S=M_x\Phi,$$
where $M_x$ is the operator of multiplication by the identity function $x$ on $L^2(\mu_c)$, $(M_xf)(x)=xf(x)$.

\end{lemma}

\begin{proof}
The proof in \cite{Voi85} involves some heavy machinery (the Helton-Howe formula \cite{HeHo73}), so we will give here a more elementary proof. 

First define the isometric isomorphism $J_1:l^2(\bn_0)\rightarrow L^2([0,\pi],\frac{2}{\pi}\,dx)$,
$$(J_1\delta_k)(x):=\sin(k+1)x,\quad (x\in[0,\pi],k\geq0).$$
Elementary Fourier theory (for odd functions on $[-\pi,\pi]$) shows that this is an isometric isomorphism. We claim that
\begin{equation}
	J_1\preal S=M_{\cos x}J_1
	\label{eqj1}
\end{equation}
where $M_{\cos x}$ is the operator of multiplication by $\cos x$ on $L^2([0,\pi])$. We have
$$J_1\preal S\delta_0=J_1\frac12\delta_1=\frac12\sin(2x)=\cos x\sin x=M_{\cos x}J_1\delta_0.$$
For $k>0$:
$$J_1\preal S\delta_k=J_1\frac12(\delta_{k+1}+\delta_{k-1})=\frac12(\sin (k+2)x+\sin kx)=\sin(k+1)x \cos x=M_{\cos x}J_1\delta_k.$$

Since $J_1\delta_0=\sin x$, we define $J_2:L^2([0,\pi])\rightarrow L^2([0,\pi],\frac{2}{\pi}\sin^2x\,dx)$ 
$$(J_2f)(x)=f(x)\frac{1}{\sin x},\quad(x\in[0,\pi]).$$
\end{proof}
Then $J_2(J_1\delta_0)=1$, the constant function $1$. And clearly $M_{\cos x}J_2=J_2M_{\cos x}$, and $J_2$ is an isometric isomorphism.

Finally, define the change of variable operator $J_3:L^2([0,\pi],\frac2\pi\sin^2 x\,dx)\rightarrow L^2([-1,1],\mu_c)$
$$(J_3f)(x)=f(\cos^{-1}x),\quad(x\in[-1,1]).$$
Then $J_3M_{\cos x}=M_xJ_3$. To check that $J_3$ is an isometry, use the change of variable $\cos^{-1} x=y$, i.e. $x=\cos y$:
$$\int_{0}^{\pi}f(y)\sin^2 y\,dy=\int_{1}^{-1}f(\cos^{-1}x)(1-x^2)\frac{-1}{\sqrt{1-x^2}}\,dx=\int_{-1}^1f(\cos^{-1}x)\sqrt{1-x^2}\,dx.$$
$J_3$ is also bijective, because its inverse can be explicitly computed: change back the variable.

The desired isometric isomorphism is $\Phi:=J_3\circ J_2\circ J_1$.

\begin{definition}\label{deflapm}
Let $\mu_c$ the semicircle measure as in Lemma \ref{lemvoic}. Define the following operators on $L^2(\mu_c)$:
$$E(f)=\int_{-1}^1f(x)\,d\mu_c(x)=\frac{2}{\pi}\int_{-1}^1f(x)\sqrt{1-x^2}\,dx,(\mbox{ considered as a constant function })\quad(f\in L^2(\mu_c)).$$
$$(A_\Delta f)(x)=2\sqrt{N}xf(x)+E(f)\quad(x\in[-1,1],f\in L^2(\mu_c).$$
$$(M_xf)(x)=xf(x),\quad(x\in[-1,1],f\in L^2(\mu_c)).$$
\end{definition}

\begin{theorem}\label{thlapm}
For $N\geq 2$, the Laplacian operator $\Delta$ is unitarily equivalent to the operator 
$$((N+1)I-A_\Delta)\oplus\oplus_{n\in\bn} ((N+1)I-2\sqrt{N}M_x),\mbox{ on the Hilbert space } \oplus_{n\geq 0}L^2(\mu_c).$$

For $N=1$, the Laplacian operator $\Delta$ is unitarily equivalent to 
$$(N+1)I-A_\Delta\mbox{ on } L^2(\mu_c).$$
\end{theorem}

\begin{proof}
The conclusion follows directly from Theorem \ref{propcycl} and Lemma \ref{lemvoic}. The only thing that remains to be proved is that the projection $P_{\delta_0}$ in $l^2(\bn_0)$ is mapped onto the operator $E$ on $L^2(\mu_c)$ by the unitary equivalence in Lemma \ref{lemvoic}. But this is clear since this unitary maps $\delta_0$ into the constant function $1$.  
\end{proof}

\begin{proposition}\label{propspec}
The spectrum of the Laplacian $\Delta$ is $[N+1-2\sqrt{N},N+1+2\sqrt{N}]$. The Laplacian $\Delta$ has no eigenvalues. 
\end{proposition}

\begin{proof}
By Theorem \ref{thlapm}, it is enough to analyze the operators $A_\Delta$ and $2\sqrt{N}M_x$ as in Definition \ref{deflapm}.
The spectrum of $2\sqrt{N}M_x$ is $[-2\sqrt{N},2\sqrt{N}]$. Since the measure $\mu_c$ has no atoms, the operator $2\sqrt{N}M_x$ has no eigenvalues. 

We turn now to the operator $A_\Delta=2\sqrt{N}M_x+E$. 
Since $A_\Delta$ is selfadjoint, the spectrum is contained in $\br$.

We will need the following
\begin{equation}
	I_\lambda:=\frac{2}{\pi}\int_{-1}^1\frac{1}{\lambda-2\sqrt{N}x}\,d\mu_c\neq1,\quad\mbox{ if }|\lambda|>2\sqrt{N}, \lambda\in\br.
	\label{eqint}
\end{equation}
To prove \eqref{eqint}, note that the function $\lambda\mapsto I_\lambda$ is decreasing on the intervals $(-\infty,-2\sqrt{N}]$ and $[2\sqrt{N},\infty)$. Also, by a change of variable $x=-y$ we obtain that $I_{-\lambda}=-I_\lambda$.

Therefore it is enough to compute 
$$I_{2\sqrt{N}}=\frac{2}{\pi}\int_{-1}^1\frac{1-x^2}{2\sqrt{N}-2\sqrt{N}x}\,dx=\frac{1}{\pi\sqrt{N}}\int_{-1}^1\sqrt{\frac{1+x}{1-x}}\,dx= (\mbox{ use the substitution}\sqrt{\frac{1+x}{1-x}}=u)$$
$$=\frac{4}{\pi\sqrt{N}}\int_0^\infty\frac{u^2}{(u^2+1)^2}\,du= (\mbox{ use the substitution }u=\tan t) =\frac{4}{\pi\sqrt{N}}\int_{0}^{\pi/2}\sin^2 t\,dt=\frac{1}{\sqrt{N}}\leq 1.$$

This implies \eqref{eqint}.

We prove that for $|\lambda|>2\sqrt{N}$ the operator $\lambda I-A_\Delta$ has a bounded inverse.
Let $g\in L^2(\mu_c)$. We want to solve 
$$\lambda f(x)-2\sqrt{N}xf(x)-E(f)=g(x),\quad(x\in[-1,1]).$$
Equivalently
\begin{equation}
f(x)(\lambda-2\sqrt{N}x)-E(f)= g(x)	
	\label{eqf1}
\end{equation}

Integrating with respect to $\mu_c$ we obtain 
$$E(f)(1-I_\lambda)=\int_{-1}^1\frac{g(x)}{\lambda-2\sqrt{N}x}\,d\mu_c(x).$$
With \eqref{eqint}
\begin{equation}
	E(f)=\frac{1}{1-I_\lambda}E\left(\frac{g(x)}{\lambda-2\sqrt{N}x}\right).
	\label{eqf2}
\end{equation}

Then, from \eqref{eqf1} and \eqref{eqf2} we get
\begin{equation}
	f(x)=\frac{g(x)+E(f)}{\lambda-2\sqrt{N}x}.
	\label{eqf3}
\end{equation}
Since $|\lambda|>2\sqrt{N}$ the function $\frac{1}{\lambda-2\sqrt{N}x}$ is bounded on $[-1,1]$. Therefore, from \eqref{eqf2}, using H\"older's inequality 
$\|E(f)\|_2\leq C\|g\|_2$, with $C$ depending only on $\lambda$. Then, from \eqref{eqf3}, $\|f\|_2\leq C'(\|g\|_2+\|E(f)\|_2)\leq C''\|g\|_2$, with $C''$ depending only on $\lambda$. This shows that the operator $\lambda I-A_\Delta$ has a bounded inverse.

Next we take $\lambda\in(-2\sqrt{N},2\sqrt{N})$ and we show that the operator $\lambda I-A_\Delta$ is not onto. Take 
$$g(x)=\left\{\begin{array}{cc} \lambda-2\sqrt{N}x,&\mbox{ if }|x-\frac{\lambda}{2\sqrt N}|>\frac14,\\
0,&\mbox{ if }|x-\frac{\lambda}{2\sqrt N}|\leq\frac14.\end{array}\right.$$
Clearly $g$ is in $L^2(\mu_c)$. Suppose there exists $f\in L^2(\mu_c)$ such that $\lambda f-A_\Delta f=g$. Then, as in \eqref{eqf3}, 
\begin{equation}\label{eqf4}
f-\frac{g}{\lambda-2\sqrt{N}x}=\frac{E(f)}{\lambda-2\sqrt{N}x},\quad(x\in[-1,1]).
\end{equation}
But the two functions on the left are in $L^2(\mu_c)$ (since $g$ is zero around the singularity $x=\lambda/(2\sqrt{N})$), while the one on the right is not, unless $E(f)=0$. But if $E(f)=0$ then \eqref{eqf4} implies that 
$f(x)=0$ around $x=\lambda/(2\sqrt{N})$ and $f(x)=1$ otherwise. Then $E(f)>0$, a contradiction.

Thus the spectrum contains $(-2\sqrt{N},2\sqrt{N})$ and since it is closed, it follows that it is equal to $[-2\sqrt{N},2\sqrt{N}]$.

Also, if $\lambda$ is an eigenvalue for $A_\Delta$, then as we have seen above, $|\lambda|\leq 2\sqrt{N}$. If $|\lambda|\leq2\sqrt{N}$, and $f$ is an eigenvector, then 
$$f(x)=\frac{E(f)}{\lambda-2\sqrt{N}x},\quad(x\in[-1,1]).$$
But the function on the right is not $\mu_c$-square integrable (at $x=\lambda/(2\sqrt{N})$), unless $E(f)=0$, in which case $f\equiv 0$. Thus there are no eigenvalues.

\end{proof}

\subsection{Rank-one perturbations}\label{rank}

 Since our graph considerations are global in nature, on the face of things, it may seem surprising that the spectral theory of associated graph Laplacians will involve rank-one perturbations of selfadjoint operators; usually thought of as ``local''. Similarly, one might think that such perturbations might be ``harmless'', but nonetheless they can be "drastic" from a spectral theoretic viewpoint; and they are intimately tied in with such deep theories as Krein's spectral shift formula, Aronszajn-Donoghue's theory for rank-one perturbations, and the generation of singular continuous spectra. Rank-one perturbations further explain resonances and spectral-shift near Landau levels in atomic physics. They can generate large point spectrum, as well as singular continuous spectrum, as measured by a computation of Hausdorff dimensions. The reader may get a sense of these intricacies from following references \cite{AKK04, AKK05, ASM07, BBR07, DSS07, KuWa01, Pol98, DJLS96, DKS06}.

\begin{theorem}\label{thmucp}
Define the measure $\mu_{c+p}$ on $[-1,1]$ by
\begin{equation}
	d\mu_{c+p}:=\frac{\frac{2}{\pi}\sqrt{1-x^2}}{1-2N^{-1/2} x+N^{-1}}\,dx.
	\label{eqmucp}
\end{equation}

The restriction of the operator $\W\Delta\W^*$ to $\H_\Omega$ (see Definition \ref{defh_i}, Proposition \ref{proplapf} and Theorem \ref{propcycl}) is unitarily equivalent to the operator of multiplication by $N+1-2\sqrt{N}x$ on $L^2(\mu_{c+p})$.
\end{theorem}

\begin{remark}
 Our plan is to make use of the moments of the operator $\preal S$ from Lemma
\ref{lemvoic}. This refers to the cyclic vector introduced in the lemma.
   As we will see, \eqref{eqc2} below, up to a geometric factor the moment numbers
are the Catalan numbers. Since the generating function for the Catalan
numbers is known, we are able to use this in analyzing the spectral picture
for our particular rank-one perturbations of $\preal S$.
   Details: We begin with a fundamental principle (Theorem \ref{thbt}) in spectral
theory. The idea behind this is based on a fundamental idea dating back to
Marshall Stone \cite{Sto32}, and valid generally for the spectrum of selfadjoint
operators $A$ in Hilbert space: One is interested in the spectral measure of a
given selfadjoint operator $A$, and its support (= the spectrum of $A$.) Stone
\cite{Sto32} suggested (see Lemma \ref{lemr1}) that spectral data may be computed from the study of the
resolvent operator $R(z) = (A - z )^{-1}$  of $A$, thinking here of $R(z)$ as an
operator valued analytic function. Since $A$ is selfadjoint, $R(z)$ is well
defined as a bounded operator function defined in the union of the upper and
the lower halfplane. It is analytic in the two halfplanes, but fails to
continue analytically across the real axis precisely at the spectrum of $A$.
Evaluation of $R(z)$ in a state yields the Borel transform $F(z)$ of the
corresponding spectral measure (Lemma \ref{lemr1}), and the same analytic
continuation/reflection principle applies, hence Theorem \ref{thbt}.

  As we make a continuation in $z$ across a point $x$ on the real line, Stone
suggested (see also Theorem \ref{thbt}) that the singular points of the continuation of $F(\cdot)$ occur
precisely when $x$ is in the spectrum of $A$. In the absolutely continuous part
of the spectrum, the limit from the upper halfplane of the imaginary part of
$F$ is the Radon Nikodym derivative of the spectral measure.

This idea was further developed in mathematical physics under the name ``edge
of the wedge'', see e.g., \cite{Rud71}.
\end{remark}

\begin{proof}{\it ( of Theorem \ref{thmucp})}
By Theorem \ref{propcycl}, the operator is unitarily equivalent to $(N+1)I-2\sqrt{N}(\preal S+\frac{1}{2\sqrt{N}}P_{\delta_0})$. We know the spectral picture for the operator $\preal S$. It is given by the semicircular law in Lemma \ref{lemvoic}. We are dealing here with a rank one perturbation of this operator. We follow the ideas from \cite{Sim95}.

The moments of the semicircular measure are given by the Catalan numbers (see \cite{RPB07}).

Let $C_n:=\frac{1}{n+1}\left(\begin{array}{c}2n\\n
\end{array}\right)$, be the Catalan numbers.

Then their generating function is (see \cite{Erd06}):

\begin{equation}
	\mathcal C(x):=\sum_{n\geq 0}C_nx^n=\frac{1-\sqrt{1-4x}}{2x}
	\label{eqc1}
\end{equation}
By the ratio test the series is convergent if $|x|<\frac14$.

The relation between the Catalan numbers and the moments of $\mu_c$ is (see \cite{RPB07}):
\begin{equation}
	\int_{-1}^1 x^{2n}\,d\mu_c=\frac{2}{\pi}\int_{-1}^1x^{2n}\sqrt{1-x^2}\,dx=\frac{1}{2^{2n}}C_n
	\label{eqc2}
\end{equation}

As follows from e.g., \cite{GLS07,RPB07} the stated moment relations for the Catalan numbers may be derived by induction from the following recursive relations
\begin{equation}
	C_0=1,\quad C_{k+1}=\sum_{n=0}^kC_nC_{k-n}.
	\label{eqcata1}
\end{equation}	
	In particular $$C_{k+1}=\frac{2(2k+1)}{k+2}C_k,\quad C_k\cong\frac{4^k}{k^{3/2}\sqrt{\pi}}, \quad C_k=\left(\begin{array}{c}2k\\k\end{array}\right)-\left(\begin{array}{c}2k\\k-1\end{array}\right),$$
	and
	\begin{equation}
	\int_{-R}^Rx^{2k}\frac{2}{\pi R^2}\sqrt{R^2-x^2}\,dx=\left(\frac{R}{2}\right)^{2k}C_k.
	\label{eqcata2}
\end{equation}

Note that the odd moments of the semicircle law are all zero. Formula \eqref{eqc1} follows from \eqref{eqcata1} by the following simple derivation
$$\mathcal C(x)=\sum_{n=0}^\infty C_nx^n=1+x\sum_{k=0}^\infty C_{k+1}x^k=\, (\mbox{ by \eqref{eqcata1} })\,=1+x\sum_{k=0}^\infty\sum_{n=0}^kC_nC_{k-n}x^k=1+x\mathcal C(x)^2.$$
Solving the quadratic equation for $\mathcal C(x)$,  
then yields the two solutions
$$\frac{1\pm\sqrt{1-4x}}{2x}.$$
We pick the solution from the ``-'' choice as it satisfies the correct boundary condition $\mathcal C(x)=1$ at $x=0$.

The Borel transform of a measure is 
$$F(z)=\int \frac{1}{x-z}\,d\mu(x)$$

We compute the Borel transform of the semicircular measure $\mu_c$:
$$F(z)=\int\frac{1}{x-z}\,d\mu_c(x)=\frac{-1}{z}\int_{-1}^1\frac{1}{1-\frac xz}\,d\mu_c(x)=
\frac{-1}{z}\sum_{n=0}^\infty\frac{1}{z^n}\int_{-1}^1x^n\,d\mu_c=$$
$$=\frac{-1}{z}\sum_{n=0}^\infty\frac{1}{z^{2n}}\frac{1}{2^{2n}}C_n=\frac{-1}{z}\sum_{n=0}^\infty C_n\left(\frac{1}{4z^2}\right)^n=\frac{-1}{z}\mathcal C\left(\frac{1}{4z^2}\right).$$
Thus
\begin{equation}
	F(z)=\frac{-1}{z}\mathcal C\left(\frac{1}{4z^2}\right)=-2z\left(1-\sqrt{1-\frac1{z^2}}\right)
	\label{eqc3}
\end{equation}
For the moment, we know that the relation holds for $|z|>1$, but we will show it actually holds for $z\in\bc\setminus[-1,1]$.

\begin{theorem}\label{thbt}\cite[Theorem 1.6]{Sim95}
Let $F(z)$ be the Borel transform of a measure $\mu$, obeying $\int\frac{d\lambda}{|\lambda|+1}<\infty$. Then
\begin{enumerate}
\item
The singular part of $\mu$ is supported on the set
$$\{x\,|\, \lim_{\epsilon\downarrow0}\operatorname*{Im} F(x+i\epsilon)=\infty\}.$$
\item The absolutely continuous part of $\mu$ has Radon-Nikodym derivative:
$$\frac{d\mu_{ac}}{dx}=\frac{1}{\pi} \lim_{\epsilon\downarrow0}\operatorname*{Im} F(x+i\epsilon).$$
\end{enumerate} 
\end{theorem}

Let $F_\alpha$ be the Borel transform of the spectral measure $\mu_\alpha$ with respect to the cyclic vector $1$, of the rank-one perturbation $M_x+\alpha E$ on $L^2(\mu_c)$. (For us $\alpha=\frac{1}{2\sqrt{N}}$. See Theorem \ref{thlapm}.)

The Aronszajn-Krein formula is (see \cite[Equation (1.13)]{Sim95}):
\begin{equation}
	F_\alpha(z)=\frac{F(z)}{1+\alpha F(z)}
	\label{eqc4}
\end{equation}

Thus we can explicitly compute $F_\alpha(z)$.

Then 
\newcommand{\pim}{\operatorname*{Im}}
$$\pim F_\alpha(z)=\pim \frac{ F(z)(1+\alpha \cj{F(z)}) }{|1+\alpha F(z)|^2}=\frac{\pim F(z)}{|1+\alpha F(z)|^2}.$$

By Theorem \ref{thbt}, we should have 
$\lim_{\epsilon\downarrow0} \frac{1}{\pi}\pim F(x+i\epsilon)=$ Wigner's semicircular distribution on $[-1,1]$, and $0$ outside $[-1,1]$. 

And it remains only to compute $\lim_{\epsilon\downarrow0}|1+\alpha F(x+i\epsilon)|^2$. Assuming that we get finite limits, this should be the distribution of our spectral measure of the perturbation.

Consider the branch of the square root defined on $\bc\setminus\{x\in\br\,|\,x\leq 0\}$. The square root function is holomorphic in this domain.

Take now $z\in\bc\setminus[-1,1]$. Then $z^2\in\bc\setminus[0,1]$, so $\frac{1}{z^2}\in\bc\setminus [1,\infty)$ so $1-\frac{1}{z^2}\in\bc\setminus(-\infty,0]$. Therefore the function in the righthand side of \eqref{eqc3} is analytic in $\bc\setminus[-1,1]$. Also $F$ is analytic in this domain, therefore the formula \eqref{eqc3} is valid for all $z\in\bc\setminus[-1,1]$.

The square root:

As we mentioned above, we define the square root function on $\bc\setminus(-\infty,0]$. An easy computation shows that 
\begin{equation}
	\sqrt{x+iy}=\sqrt{\frac{\sqrt{x^2+y^2}+x}{2}}+i\operatorname*{sgn}(y)\sqrt{\frac{\sqrt{x^2+y^2}-x}{2}}
	\label{eqc5}
	\end{equation}
	
	where 
	$$\operatorname*{sgn}(y)=\left\{\begin{array}{cc}1,&y>0\\
	0,&y=0\\
	-1,&y<0.\end{array}\right.$$

The next Lemma can be checked immediately.

\begin{lemma}\label{lemc1}
If $x<0$ and $z_n\rightarrow x$ then:
\begin{enumerate}
\item If $\pim z_n>0$ for all $n$ then $\sqrt{z_n}\rightarrow i\sqrt{-x}$.
\item If $\pim z_n<0$ for all $n$ then $\sqrt{z_n}\rightarrow -i\sqrt{-x}$.
\end{enumerate}
\end{lemma}

\begin{lemma}\label{lemc2}
If $x\in[-1,1]$, $z_n=x_n+iy_n\rightarrow x$, and $y_n>0$ for all $n$, then:
\begin{enumerate}
\item If $x>0$, then $\sqrt{1-\frac{1}{z_n^2}}\rightarrow i\sqrt{\frac{1}{x^2}-1}$.
\item If $x<0$, then $\sqrt{1-\frac{1}{z_n^2}}\rightarrow -i\sqrt{\frac{1}{x^2}-1}$.
\end{enumerate}
\end{lemma}

\begin{proof}
If $x>0$ then $x_n>0$ for $n$ big, and since $y_n>0$ it follows that $\pim z_n^2>0$. Then $\pim\frac{1}{z_n^2}<0$ so $\pim(1-\frac{1}{z_n^2})>0$. Then (i) follows from Lemma \ref{lemc1}.

(ii) can be obtained similarly.

\end{proof}

With Lemma \ref{lemc2}, for $x>0$, 
$$\lim_{y\downarrow 0} F(x+iy)=-2x\left(1-i\sqrt{\frac{1-x^2}{x^2}}\right)=-2x+ 2i\sqrt{1-x^2}.$$
For $x<0$
$$\lim_{y\downarrow 0} F(x+iy)=-2x\left(1+i\sqrt{\frac{1-x^2}{x^2}}\right)=-2x+2i\sqrt{1-x^2}.$$

Then,
for $x\in[-1,1]$, 
$$\lim_{y\downarrow 0}\frac{1}{\pi}\pim F_\alpha(x+iy)=
\frac{\frac{2}{\pi}\sqrt{1-x^2}}{(1-2\alpha x)^2+4\alpha^2(1-x^2)}=\frac{\frac{2}{\pi}\sqrt{1-x^2}}{1-4\alpha x+4\alpha^2}.$$

Then we plug in $\alpha=\frac{1}{2\sqrt{N}}$ and we obtain that $M_x+\frac{1}{2\sqrt{N}}E$ is unitarily equivalent to multiplication by $x$ on $L^2(\mu_{c+p})$, and the conclusion of the Theorem follows.
\end{proof}

\begin{theorem}\label{th3.25}
Let $\mu_c$ be the semicircular measure on $[-1,1]$ 
$$d\mu_c=\frac{2}{\pi}\sqrt{1-x^2}\,dx,$$
and let $\mu_{c+p}$ be the measure on $[-1,1]$ given by 
$$d\mu_c=\frac{\frac2{\pi}\sqrt{1-x^2}}{1-2N^{-1/2} x+N^{-1}}\,dx.$$
Then
\begin{enumerate}
\item If $N=1$, then the Laplacian $\Delta$ is unitarily equivalent to the operator of multiplication by $2-2x$ on $L^2(\mu_{c+p})$. 
\item IF $N\geq 2$, then the Laplacian $\Delta$ is unitarily equivalent to the multiplication operator 
$$M_{c+p}\oplus\oplus_{n=1}^\infty M_c\mbox{ on } L^2(\mu_{c+p})\oplus\oplus_{n=1}^\infty L^2(\mu_c),$$
where $M_{c+p}$ is the operator of multiplication by $N+1-2\sqrt{N}x$ on $L^2(\mu_{c+p})$, and $M_c$ is the operator of multiplication by $N+1-2\sqrt{N}x$ on $L^2(\mu_c)$.
\end{enumerate}
\end{theorem}

\begin{proof}
The Theorem follows directly from Theorem \ref{propcycl}, Lemma \ref{lemvoic}, and Theorem \ref{thmucp}.
\end{proof}

\begin{remark}
 Theorem \ref{th3.25} shows that the particular rank-one perturbations used in the spectral representation for $\Delta_G$ do not introduce point spectrum. Here $G$ refers to the $N$-ary tree. This could not have been predicted by the general theory of rank-one perturbations, see e.g., \cite{AKK04}. Nonetheless, as we show below, $\Delta_G$ has infinite-energy (i.e., not $l^2$-) eigenvectors. For the sake of simplicity, we complete the details only in special case of $N = 1$, but an analogous construction works in general.

 While the Laplace operator $\Delta_G$ for $G=$the $N$-bifurcation graph has absolutely continuous spectrum, the following example (for $N=1$) shows that $\Delta_G$ may have {\it infinite}-energy eigenvectors.

\end{remark}

\begin{example}\label{exgo1}
On sequences $u=(u_0,u_1,u_2,\dots)$, set 
\begin{equation}
	(\Delta u)_0=u_0-u_1,\mbox{ and } (\Delta u)_n=2u_n-u_{n-1}-u_{n+1},\quad n\geq1.
	\label{eqgo1}
\end{equation} 

In general $\lambda=0$ is always an eigenvalue. Indeed the vector $v=(1,1,1,\dots)$ satisfies $\Delta v=0$.

The idea in the algorithm of this example is that we generate a finite
system of eigenvalues and eigenvectors for each $n$. Given $n$, the admissible
eigenvalues $\lambda$ occur as roots in a polynomial $p_n$.

   Details: Fix $n$ and examine a vector which begins with a finite word $w$,
all letters in $w$ assumed nonzero. Letting $\lambda$ be free, and setting $v_n =
0$, you get a polynomial equation $v_n = p_n(\lambda) =0$. Now solve for
$\lambda$, and then use the recursion to generate the rest of the coordinates
in a $\lambda$-eigenvector, i.e., use the recursion to compute $v_k$ for $k =
n+1, n+2,\dots$. In each case, when $\lambda$ is fixed, we will get a periodic
sequence for the entire vector $v = (v_0, v_1,\dots )$. Or course this
periodicity implies boundedness of $(v_k)$, as $k$ varies in $\bn$. For each of the
admissible values of $\lambda$ you get a one-dimensional eigenspace. So for
convenience, we can set $v_0 = 1$.

In the special case of $N=1$, set $\xi=(1,0,-1,-1,0,1)$ and $v=(\xi,\xi,\xi,\dots)$, repetition of the finite word $\xi$, satisfies $\Delta v=v$, so $\lambda=1$ is also an infinite energy eigenvector.

Continuing this process, we can show that for every $n\in\bn$ there is a monic polynomial $p_n(\lambda)$ such that each root $\lambda$ in $p_n(\lambda)=0$ is an infinite-energy eigenvalue of $\Delta$. Moreover the sets of roots are disjoint.

Specifically,
$$p_1(\lambda)=1-\lambda,\quad p_2(\lambda)=1-3\lambda+\lambda^2$$
and
$$p_{n+1}(\lambda)=(2-\lambda)p_n(\lambda)-p_{n-1}(\lambda).$$

For the case $n=2$, the two roots are $\lambda_{\pm}=\frac{3\pm\sqrt{5}}{2}$, and we proceed with the analysis of this case. Then the eigenvalue problem 
\begin{equation}
	\Delta v=\lambda v
	\label{eqgo2}
\end{equation}
has non-zero solutions $v\neq 0$ for the two Golden ration numbers
\begin{equation}
	\lambda=\frac{3\pm\sqrt{5}}{2}
	\label{eqgo3}
\end{equation}
In each case, the eigenspace is spanned by the following two periodic sequences There are words $\xi_{\pm}$ and $\eta_{\pm}$, $\xi_\pm$ of length $2$ and $\eta_\pm$ of length $8$ respectively, such that the corresponding $\lambda_\pm$ eigenspaces are spanned by
\begin{equation}
	(\xi\eta,\xi\eta,\dots)\mbox{ infinite repetition}.
	\label{eqgo4}
\end{equation}
\begin{proof}
A computation shows that every solution $(v_0v_1v_2\dots)$ to \eqref{eqgo2} satisfies $v_1=(1-\lambda)v_0$, $v_2=(\lambda^2-3\lambda+1)v_0$. 
So if $\lambda$ is one of the roots in \eqref{eqgo3}, then there are solutions $v$ spanned by 
$v=(1,1-\lambda,0,\mbox{ some infinite word})$. Continuation of the iteration in \eqref{eqgo1} yields 
$$\begin{array}{ccccccccccccc}
(1,&1-\lambda,&0,&\lambda-1,&-1,&-1,&\lambda-1,&0,&1-\lambda,&1,&1,&1-\lambda,\dots)\\
(v_0&v_1)&(v_2&v_3&v_4&v_5,&v_6,&v_7,&v_8,&v_9),(&v_{10},&v_{11}),\dots\end{array}$$

Setting $\xi=(1,1-\lambda)$ and $\eta=(0,\lambda-1,-1,-1,\lambda-1,0,1-\lambda,1)$ the desired conclusion follows.
\end{proof}

 These infinite-energy eigenvalues may be of physical significance as
they serve to show that there is a natural way to ``renormalize'', introducing
a weighted sequence space, in such a way that in the renormalized Hilbert
space, these infinite-energy eigenfunctions turn into finite-energy.
\end{example}

\section{Resistance metric}\label{resi}

  Reviewing Definition \ref{def2.1}, formula \eqref{eq1.4} for the graph Laplacian, $\Delta_{G,c}$ the reader will note the coefficient $c$. Its significance is seen for example from electrical network models built on a given graph $G = (G^{(0)}, G^{(1)})$. Here $c$ represents conductance, and it takes the form of a positive function defined on the set $G^{(1)}$ of all edges. It signifies the reciprocal of resistance. The graph $G$ is then a ``large'' (means infinite!) system of resistors, each edge $e =(xy)$ representing a resistance of $c(e)^{-1}$  Ohm between the neighboring vertices. The question arises of determining some sort of resistance metric giving the resistance between an arbitrary pair of vertices, so an arbitrary pair of points $x$ and $y$ in $G^{(0)}$. This is non-trivial as there typically are many paths of edges connecting $x$ with $y$. Nonetheless, the rules for electrical networks, Ohm's law combined with Kirchhoff's law, allow us to compute a useful measure resistance, which we shall refer to as the resistance metric; see \cite{Jor08, Kig03, Pow76}.

     One of the conclusions in the operator approach to resistance amounts to identifying the resistance metric as a norm difference: The norm is now referring to a Hilbert space, and derived from the graph Laplacian $\Delta_{G,c}$, i.e., the operator \eqref{eq1.4} for a particular choice of conductance function $c$. Hence there is a Hilbert space, the energy Hilbert space $E= E(\Delta_{G,c})$ and a function $v$ from $G^{(0)}$ into $E$ such the resistance between $x$ and $y$ is the norm difference $\| v(x)- v(y)\|_E$, i.e., the $E$-norm difference between the vectors $v(x)$ and $v(y)$ in $E$. But each vector $v(x)$ is itself a function on the set of vertices $G^{(0)}$, in fact an electrical potential defined by $\Delta_{G,c}$; see Lemma \ref{lemre2} below.

      We further note that there is a general theory of metrics which can be computed this way with the use of Hilbert space, and we refer to the paper \cite{Fug05} for an overview. In fact, these metric embeddings form a subclass of a wider family: spirals, or screw functions, as defined by von Neumann; and they have applications in information theory. In the present context, we need them for making precise the resistance metric, and to motivate the Energy Hilbert space.

\begin{definition}\label{defre1}
Let $G=(G^{(0)},G^{(1)})$ be a graph with conductance function $c:G^{(1)}\rightarrow\br_+$. Pick a point $o$ in $G^{(0)}$. For $x\in G^{(0)}$, the equation 
\begin{equation}
	\Delta_{G,c}v_x=\delta_o-\delta_x
	\label{eqre1}
\end{equation}
has a unique solution $v_x$ in the Hilbert space $E$ obtained by completion relative to the quadratic form
\begin{equation}
	\|u\|_E^2:=\sum_{x\in G^{(0)}}\sum_{y\sim x} c(xy)|u(x)-u(y)|^2.
	\label{eqre2}
\end{equation}

The expression $\E(u',u)$ is determined by polarization from \eqref{eqre2}, i.e., 
$$\E(u',u)=\sum_{x\in G^{(0)}}\sum_{y\sim x}c(xy)(\cj u'(x)-\cj u'(y))(u(x)-u(y)).$$
The {\it resistance metric} on $G^{(0)}$ is given by
$$\operatorname*{dist}(x,y)=\|v_x-v_y\|_E,\quad(x,y\in G^{(0)}).$$

Here we are assuming that $G$ is connected. For details see \cite{Jor08}.
\end{definition}

\begin{lemma}\label{lemre2}
The solution $v_x$ in \eqref{eqre1} is determined uniquely up to an additive constant by the formula
$$\frac{1}{2}\E(v_x,u)=u(o)-u(x),\quad(u\in\mathcal D).$$
\end{lemma}
\begin{proof}
We check that
$$\E(v,u)=2\ip{\Delta_{G,c}v}{u}_{l^2(G^{(0)})}.$$
Hence for all $u\in\mathcal D$, we have 
$$\frac12\E(v_x,u)=\ip{\Delta v_x}{u}_{l^2}= (\mbox{ by \eqref{eqre1} })\ip{\delta_o-\delta_x}{u}_{l^2}=u(o)-u(x).$$ 
\end{proof}

      Below we compute the resistance metric for the $N$-adic graph, and for the operator $\Delta_{G,c}$ from Proposition \ref{proplapf} and Theorem \ref{thmucp}.

\begin{remark}\label{remo1}
If $c:G^{(1)}\rightarrow\br_+$ is a conductance function and if $v:G^{(0)}\rightarrow \br$ is a function on the vertices of $G$, then (Ohm's law!)
\begin{equation}
	I(xy):=c(xy)(v(x)-v(y))
	\label{eqo1}
\end{equation}
defines a current flow on $G$.

For a fixed pair of points $x$ and $y$ in $G^{(0)}$, we are interested in the following experiment which inserts one Amp at $x$ and extracts it at $y$: It induces a current flow $I:G^{(1)}\rightarrow\br$ which is governed by (Kirchoff's laws)
\begin{equation}
	\sum_{G^{(0)}\ni k\sim j}I(jk)=\delta_x(j)-\delta_y(j)
	\label{eqo2}
\end{equation}

\end{remark}

\begin{lemma}\label{lemo2}
Let $x,y\in G^{(0)}$ as above. If $I:G^{(1)}\rightarrow\br$ is a solution to \eqref{eqo2}, and a voltage potential is determined from \eqref{eqo1}, then 
\begin{equation}
	\E_c(v)=\sum_{e\in G^{(1)}}\frac{1}{c(e)}(I(e))^2;
	\label{eqo3}
\end{equation}
and 
\begin{equation}
	\Delta_{G,c}v=\delta_x-\delta_y.
	\label{eqo4}
\end{equation}

\end{lemma} 
\begin{proof}
First, \eqref{eqo3} follows from a substitution of \eqref{eqo1} into \eqref{eqre2}. For \eqref{eqo4}, we get 
$$(\Delta_{G,c}v)(j)=\sum_{k\sim j}c(jk)(v(j)-v(k))=(\mbox{ by \eqref{eqo1} } )\sum_{k\sim j}I(jk)=(\mbox{ by \eqref{eqo2} })=\delta_x(j)-\delta_y(j).$$
\end{proof}
\begin{proposition}\label{prop4.5}
Consider the $N$-ary tree $G=(G^{(0)},G^{(1)})$ in Definition \ref{def3.1}.
Let $\eta:=\eta_1\dots\eta_n$ be a word in $G^{(0)}$, $n\geq1$. Then the solution (potential) $v\in E$ to 
\begin{equation}
	\Delta v=\delta_\ty-\delta_{\eta_1\dots\eta_n}
	\label{eqpot1}
\end{equation} 
is given by
\begin{equation}
	v(\omega_1\dots\omega_m)=n-p(\omega_1\dots\omega_m,\eta_1\dots\eta_n)
	\label{eqpot2}
\end{equation}
where $p(\omega_1\dots\omega_m,\eta_1\dots\eta_m)$ is the length of the largest common prefix for $\omega_1\dots\omega_m$ and $\eta_1\dots\eta_n$, i.e., the largest $p\geq 0$ such that $p\leq m,n$ and 
$\omega_i=\eta_i$ for all $i=1,\dots,p$.

The resistance metric is 
\begin{equation}
	\operatorname*{dist}(x,y)=\sqrt{2l(x,y)},
	\label{eqpot3}
	\end{equation}
	where $l(x,y)$ is the length of the shortest path from $x$ to $y$.

\end{proposition}
\begin{proof}
We check \eqref{eqpot1}. 
$$(\Delta v)(\ty)=Nv(\ty)-\sum_{i=1}^Nv(i)=Nn-\sum_{i\neq\eta_1}v(i)-v(\eta_1)=
Nn-(N-1)n-(n-1)=1=\delta_\ty(\ty)-\delta_\eta(\ty).$$
Next, we check the prefixes $\eta_1\dots\eta_p$ with $p\leq n-1$
$$(\Delta v)(\eta_1\dots\eta_p)=(N+1)v(\eta_1\dots\eta_p)-\sum_{i\neq \eta_{p+1}}v(\eta_1\dots\eta_p i)-v(\eta_1\dots\eta_{p+1})-v(\eta_1\dots\eta_{p-1})=$$$$
(N+1)(n-p)-(N-1)(n-p)-(n-p-1)-(n-p+1)=0=\delta_\ty(\eta_1\dots\eta_p)-\delta_\eta(\eta_1\dots\eta_p).$$
$$(\Delta v)(\eta_1\dots\eta_n)=(N+1)v(\eta_1\dots\eta_n)-\sum_{i=1}^Nv(\eta_1\dots\eta_ni)-v(\eta_1\dots\eta_{n-1})=$$$$(N+1)\cdot0-N\cdot 0-1=-1=\delta_\ty(\eta_1\dots\eta_n)-\delta_\eta(\eta_1\dots\eta_n).$$

All the other words have the form $\omega:=\eta_1\dots\eta_p\omega_{p+1}\dots\omega_m$ for some $0\leq p\leq n$, $\omega_{p+1}\neq \eta_{p+1}$ (if $p=n$, then we just take $m>n$). Since the function $v$ is constant $n-p$ on the subtree with root $\eta_1\dots\eta_p\omega_{p+1}$, all the terms in the definition of $(\Delta v)(\eta_1\dots\eta_p\omega_{p+1}\dots\omega_m)$ are equal $n-p$, and they sum to $(N+1)-N-1=0=\delta_\ty(\omega)-\delta_\eta(\omega)$.

Clearly, only a finite number of edges $(xy)$ have $v(x)-v(y)\neq 0$, namely the ones of the form $\eta_1\dots\eta_p\sim\eta_1\dots\eta_{p+1}$. Therefore $\E(v,v)<\infty$. 
This proves \eqref{eqpot2}.

To prove \eqref{eqpot3} we use \cite[Proposition 5.15]{Jor08}. Let $p$ be the length of the longest common prefix of $x$ and $y$, and let $n,m$ be the lengths of $x$ and $y$ respectively.

By \cite[Proposition 5.15]{Jor08}
$$\operatorname*{dist}(x,y)=\sqrt{2}(v_x(y)+v_y(x)-v_x(x)-v_y(y))^{\frac12},$$
where $v_x$ is the solution (potential) for $\Delta v_x=\delta_\ty-\delta_x$ and 
$v_y$ is the solution (potential) for $\Delta v_y=\delta_\ty-\delta_y$.
Then 
$$\operatorname*{dist}(x,y)=\sqrt{2}(n-p+m-p-0-0)^{\frac12}=\sqrt{2}l(x,y)^{\frac12}.$$

\end{proof}

 The following corollaries involve two general issues for the resistance metric on graphs. They are motivated by two fundamental Hilbert space constructions: The first (von Neumann and Schoenberg) is a necessary and sufficient condition on a metric space $(V, d)$ for the metric $d$ to be the restriction of a Hilbert-norm difference. This then yields an isometric embedding of $V$ into some Hilbert space. In the present case, this Hilbert space will be concrete: The energy Hilbert space $E$ in Definition \ref{defre1} and Lemma \ref{lemre2}. The idea goes back to von Neumann and Schoenberg: The a priori condition on the metric $d$ is that $d^2$ is negative semidefinite. The second construction amounts to Kolmogorov's consistency rules \cite{PaSc72}. With this, we get a Gaussian stochastic process indexed by $G^{(0)}$. Kolmogorov \cite{PaSc72}: An a priori given covariance function is positive semidefinite if and only if it comes from a stochastic Gaussian process. Below, we apply this to the energy-inner product $\ip{v_x}{v_y}_E$ where for each $x$ in $G^{(0)}$, $v_x$ is the potential function from Lemma \ref{lemre2}, and $\ip{\cdot}{\cdot}_E$ is the inner product in the Hilbert space \eqref{eqre2} in Definition \ref{defre1}.

\begin{corollary}\label{coko1}
Let $N\in\bn$ and let $G=(G^{(0)},G^{(1)})$ be the tree of Definition \ref{def3.1}. For $x,y\in G^{(0)}$, let $l(x,y)$ denote the length of the shortest path from $x$ to $y$. Then for all finitely supported functions $\xi\in G^{(0)}\rightarrow\bc$ satisfying $\sum_x\xi(x)=0$, we have
\begin{equation}
	\sum_x\sum_y\cj\xi(x) l(x,y)\xi(y)\leq 0
	\label{eqko1}
\end{equation}
(In short the function $l(\cdot,\cdot)^2$ is negative semi-definite.)
\end{corollary}

\begin{proof}
Combining Lemma \ref{lemo2} and Proposition \ref{prop4.5}, we see that 
\begin{equation}
	l(x,y)=\frac{1}{2}\operatorname*{dist}(x,y)^2=\frac{1}{2}\|v_x-v_y\|_E^2
	\label{eqko2}
\end{equation}
where $v:G^{(0)}\rightarrow E=$(the energy Hilbert space), $v(x):=v_x$ is the function determined in Lemma \ref{lemre2}. But it is known that the property in \eqref{eqko1}, (negative semi-definite) characterizes those metrics which have isometric embeddings into Hilbert space; see \cite{Fug05}. 
\end{proof}

\begin{remark}\label{remko2}
The function $v:G^{(0)}\rightarrow E=$(the energy Hilbert space) from \eqref{eqko2} also determines a system of covariances with the use of the inner products 
$$\ip{v_x}{v_y}_E=2\ip{\Delta v_x}{v_y}_{l^2}=2\ip{\delta_o-\delta_x}{v_y}_{l^2}=2(v_y(o)-v_y(x)).$$
An application of Kolmogorov's theorem \cite{PaSc72} yields a Gaussian stochastic process 
$w:G^{(0)}\rightarrow L^2(\Omega,P)$ such that
\begin{equation}
	\ip{w_x}{w_y}_E=\int_{\Omega}w_x(\omega)w_y(\omega)\,dP(\omega),\quad(x,y\in G^{(0)});
	\label{eqko3}
\end{equation}
and we may normalize with 
\begin{equation}
	\int_{\Omega}w_x(\omega)\,dP(\omega)=0,\quad(x\in G^{(0)}).
	\label{eqko4}
\end{equation}
\end{remark}

\begin{corollary}\label{corko3}
Let $N\in\bn$ and let $G=(G^{(0)},G^{(1)})$ be the $N$-ary tree of Definition \ref{def3.1}. Then the stochastic process in \eqref{eqko3} has independent increments: specifically, if three points $x,y,$ and $z$ in $G^{(0)}$ satisfy $x\leq y\leq z$ relative to the natural (partial) order of $G^{(0)}$ then
\begin{equation}
	\E(v_x-v_y,v_y-v_z)=0.
	\label{eqko5}
\end{equation}
\end{corollary}

\begin{proof}
The order relation for the three points means that there are words $r,s$ such that $y=xr$ and $z=ys$, where we use concatenation of finite words. 

Using Lemma \ref{lemre2}, we now see that \eqref{eqko5} follows from the following identity
$$v_y(y)+v_z(x)-v_y(x)-v_z(y)=0$$
which in turn follows from \eqref{eqpot2} in Proposition \ref{prop4.5}.
\end{proof}

\subsection{A random walk}
In Theorem \ref{thti1} we will compute the moments of the measure $\mu_{c+p}$ from Theorem \ref{th3.25}. For this we will use the following Proposition:

\begin{proposition}\label{propsu1}
Let $G=(G^{(0)},G^{(1)})$ be a graph and $c:G^{(1)}\rightarrow \br_+$ be a conductance function. We assume that $G$ is not oriented, but it may have loops at some vertices, i.e., edges of the form $(xx)$.  Assume in addition that 
\begin{equation}
	\sum_{y\sim x}c(xy)=1,\quad(x\in G^{(0)}).
	\label{eqsu1}
\end{equation}
Define the operator $\M$ on $\D$ by
\begin{equation}
	(\M u)(x)=\sum_{y\sim x}c(xy)u(y),\quad(x\in G^{(0)},u\in\D)
	\label{eqsu2}
\end{equation}
Then $\Delta_{G,c}=I_{l^2}-\M$.

Define the random walk on $G$ by assigning the transition from $x$ to $y$, $y\sim x$ the probability $c(xy)$. For $x,y\in G^{(0)}$ and $n\in\bn_0$, let $p(x,y;n)$ be the probability of transition from $x$ to $y$ in $n$ steps. Then

\begin{equation}
	\M^n\delta_x=\sum_{y\in G^{(0)}}p(x,y;n)\delta_y,\quad(x\in G^{(0)}).
	\label{eqsu3}
\end{equation}
Let $o\in G^{(0)}$ be a fixed vertex. Let $\mu_{\M}$ be the spectral measure associated to the operator $\M$ and the vector $\delta_o$. Then 
\begin{equation}
\int x^n\,d\mu_{\M}=\ip{\delta_o}{\M^n\delta_o}=p(o,o;n)\mbox{ the probability of return to }o\mbox{ after }n\mbox{ steps}. 	
	\label{eqsu4}
\end{equation}

\end{proposition}

\begin{proof}

From the definition of $\M$ we get that 
\begin{equation}
	\M\delta_x=\sum_{y\sim x}c(xy)\delta_x,\quad(x\in G^{(0)}).
	\label{eqsu5}
\end{equation}
Equation \eqref{eqsu3} follows by induction: for $n=1$, it is just the definition of $\M\delta_x$. Assume the equation is true for $n$, then 
$$\M^{n+1}\delta_x=\sum_yp(x,y;n)\sum_{z\sim y}c(yz)\delta_z=\sum_{z}(\sum_{y\sim z}p(x,y;n)c(yz))\delta_z=\sum_z p(x,y;n+1)\delta_z.$$
Equation \eqref{eqsu4} follows directly from \eqref{eqsu3}.
\end{proof}

\begin{definition}\label{defti1}
Let $T$ be the $N$-ary tree from Definition \ref{def3.1}. We define the graph $\tilde T$ by considering the same vertices and adding an edge from $\ty$ to itself.
\end{definition}

\begin{figure}[h]
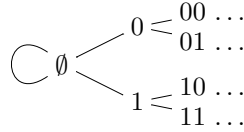

\caption{The graph $\tilde T$ ($N=2$).}\label{fig2}
\[
\xygraph{
!{<0cm,0cm>;<1cm,0cm>:<0cm,1cm>::}
!{(0,0) }*+{\ty}="a"
!{(1,0.5) }*+{0}="b"
!{(1,-0.5) }*+{1}="c"
!{(2,0.7)}*+{00\,\dots}="d"
!{(2,0.3)}*+{01\,\dots}="e"
!{(2,-0.3)}*+{10\,\dots}="f"
!{(2,-0.7)}*+{11\,\dots}="g"
"a"-"b" "a"-"c" "b"-"d" "b"-"e" "c"-"f" "c"-"g"
"a"- @(lu,ld) "a"
}
\]
\end{figure}

\begin{theorem}\label{thti1}

For each $n\in\bn$, let $N_{\tilde T}(n)=$ the number of paths of length $n$ in $\tilde T$ from $\ty$ to itself.
Then
\begin{equation}
	\int x^n\,d\mu_{c+p}=\frac{1}{(2\sqrt{N})^n} N_{\tilde T}(n),\quad(n\in\bn).
	\label{eqti1}
\end{equation}
\end{theorem}

\begin{proof}
Let $\Delta_T$ be the Laplacian associated to the tree $T$ with conductance constant $1$ on every edge. Let $\M_{\tilde T}$ be the operator defined in equation \eqref{eqsu2} for the graph $\tilde T$ with conductance $\frac{1}{N+1}$ on each edge. 
Then a simple computation shows that 
\begin{equation}
	\Delta_{T}=(N+1)I-(N+1)\M_{\tilde T}.
	\label{eqti2}
\end{equation}

From Theorem \ref{th3.25} we know that the $\Delta_{T}$ restricted to the cyclic subspace generated by $\delta_\ty$ is unitarily equivalent to an operator of multiplication by $N+1-2\sqrt{N}x$ on $L^2(\mu_{c+p})$. Therefore the restriction of $\M_{\tilde T}$ to this cyclic subspace is unitarily equivalent to an operator of multiplication by $\frac{2\sqrt{N}x}{N+1}$. Using now Proposition \ref{propsu1}, we obtain that 
$$\int \left(\frac{2\sqrt{N}x}{N+1}\right)^n\,d\mu_{c+p}=\mbox{ probability of return to }\ty\mbox{ after }n\mbox{ steps, in the graph }\tilde T.$$
But this probability is equal to $\frac{N_{\tilde T}(n)}{(N+1)^n}$ and the result follows.
\end{proof}

{\bf Concluding remarks.}
We give the complete spectral picture for graph Laplacians for a number of classes of infinite graphs. Our conclusions include spectral representation, spectral measures, multiplicity tables, occurrence of semicircle laws, and rank-one perturbations. As corollaries of our main results, we note in particular that the spectrum must necessarily be absolutely continuous for the graph Laplacians $\Delta_G$ (with normalized conductance) when $G$ is an infinite tree graph, and when $G$ is a lattice graph. It would be interesting to delimitate those infinite graphs $G$ and associated conductance functions $c$ for which the graph Laplacians $\Delta_{G,c}$  have absolutely continuous spectrum. But it seems unlikely that a complete spectral representation is available for the most general infinite graph.

A class of graphs between two extremes is those generated by tree-state automatons. 
They are studied in \cite{GrZu02}, and some spectral data is known. The class includes the 
free group Cayley graphs as well.

We have explored two such graphs: (a) the graph of free group on two generators, and (b) those with vertices on Bethe lattices \cite{Hug96}. Here the situation is more subtle: Case (a) is more complicated than the tree graphs in section 3 above. In section 3 we could build the spectral picture of $\Delta$ up with the use of a semicircle transform and a rank-one perturbation. But to get the spectral picture for $\Delta$ associated with the free groups, we might have to perturb by a rank-two operator, and the vector used for the perturbation might not be the obvious vacuum vector.

The co-authors thank many colleagues for suggested improvements on an earlier version of this paper. The second named author (PJ) learned of these graph problems from Robert T. Powers in the 1970ties. The applications from that time include electrical networks of resistors, and KMS states in statistical mechanics. Since this, the interest in discrete analysis and graph Laplacians has grown tremendously. As a result, there is now an enormous literature, which in turn branches off in a variety of different directions. We cannot begin to do justice to this vast and diverse literature. It includes discrete Schr\" odinger operators in physics, information theory, potential theory, uses of the graphs in scaling-analysis of fractals (constructed from infinite graphs), probability and heat equations on infinite graphs, graph $C^*$-algebras, groupoids, Perron-Frobenius-transfer operators (used in models for the internet); multiscale theory, renormalization, and operator theory of boundaries of infinite graphs (current joint work between PJ and Erin Pearse.) We learned of some more such directions after the completion of this paper from groups of experts, and we are thankful to many colleagues who took the time to explain them to us.

\begin{acknowledgements}
We thank our colleagues Gabriel Picioroaga, Erin Pearse, Keri Kornelson, Karen
Shuman, Qiyu Sun, and Myung-Sin Song. The second named author thanks
University of Central Florida for hospitality during a research visit where
most of this work was done. Professors Strichartz and Golenia kindly helped us with literature citations after we
completed the first version of our paper.
\end{acknowledgements}

\bibliographystyle{alpha}
\bibliography{laptr}

\end{document}